\documentclass[12pt,draftclsnofoot,onecolumn]{IEEEtran}

\newcommand\norm[1]{\left\lVert#1\right\rVert}

\usepackage{graphicx}
\ifCLASSINFOpdf
\else
\fi

\makeatletter
\newcommand*{\rom}[1]{\expandafter\@slowromancap\romannumeral #1@}
\makeatother
\usepackage{float}
\usepackage{mathtools}
\DeclarePairedDelimiter\ceil{\lceil}{\rceil}

\usepackage{subcaption} 
\usepackage{amsmath}
\usepackage{amssymb}
\usepackage{amsthm}
\usepackage{amsfonts}
\usepackage{stackengine}
\usepackage{xcolor}
\usepackage[noadjust]{cite}
\def\delequal{\mathrel{\ensurestackMath{\stackon[1pt]{=}{\scriptstyle\Delta}}}}
\newtheorem{theorem}{Theorem}

\theoremstyle{definition}
\newtheorem{definition}{Definition}
\newtheorem{prop}{Proposition} 
\theoremstyle{remark}
\newtheorem{remark}{Remark}

\begin{document}

\title{A New Analytical Approximation of the Fluid Antenna System Channel}

\author{\IEEEauthorblockN{ Malek Khammassi\IEEEauthorrefmark{1},
Abla Kammoun\IEEEauthorrefmark{1},
Mohamed-Slim Alouini\IEEEauthorrefmark{1}}\\
\IEEEauthorblockA{\IEEEauthorrefmark{1}
King Abdullah University of Science and Technology,
Thuwal, Saudi Arabia} \\
\IEEEauthorblockA{
Email: \{malek.khammassi, abla.kammoun, slim.alouini\}@kaust.edu.sa}
\thanks{Due to space limitation, most of the proofs of theorems, propositions and corollaries have been omitted and they can be found in our technical report \cite{meeeeeee}.}
}

        

%
%



\maketitle

\begin{abstract}

Fluid antenna systems (FAS) are an emerging technology that promises a significant diversity gain even in the smallest spaces. It consists of a freely moving antenna in a small linear space to pick up the strongest received signal. Previous works in the literature provide a simple yet insightful parameterization of the FAS channel that leads to single-integral expressions of the probability of outage and various insights on the achievable performance. Nevertheless, this channel model may not accurately capture the correlation between the FAS ports, given by Jake's model. This work builds on the state-of-the-art by incorporating more parameters into the channel model to accurately approximate the FAS channel distribution while maintaining analytical tractability. The approximation is performed in two stages. The first stage approximation considerably reduces the number of multi-fold integrals in the probability of outage expression, while the second stage approximation represents it in a single integral form. Numerical results validate our approximations of the FAS channel model and demonstrate a limited performance gain under a more accurate correlation model. Further, our work opens the door for future research to investigate scenarios in which the FAS provides a performance gain compared to the current multiple antenna solutions.

\end{abstract}


\IEEEpeerreviewmaketitle

\section{Introduction}

\IEEEPARstart{F}{or} the past few decades, multiple-input multiple-output (MIMO) has been one of the most celebrated wireless communication technologies. The philosophy behind MIMO consists of exploiting multipath, which, for very long, has been considered undesirable, to multiply the capacity. Although the earliest ideas relating capacity gain to multipath were hard to accept, MIMO has shown an undeniable performance gain, and therefore, it has become an essential component of wireless communication standards. MIMO allows data multiplexing over channels undergoing independent fading. However, the rule of thumb is to spatially separate the antennas by at least half the radiation wavelength to ensure diversity gain. 

The authors in \cite{FAS1, secondorderstat, fama, portselection} have recently questioned this rule. Motivated by the recent trend of using liquid metals such as Galinstan and ionized solutions such as sodium chloride for antennas \cite{liq1,liq2,liq3,liq4, liq5,liq6}, Wong \textit{et al.} hypothesize a system where a single antenna can switch locations instantly in a small linear space and refer to it as a fluid antenna system (FAS). They refer to the possible positions of the antenna as ports, and they investigate a scenario where the antenna can switch to the port with the strongest signal in the manner of traditional selection diversity. This system can ensure the implementation of multiple antennas at the receiver's side without space limitations. The analysis of the first-order statistics of FAS in \cite{FAS1} shows a probability of outage that decreases as the number of ports increases. Furthermore, it shows that FAS can outperform maximal-ratio combining (MRC), for a large enough number of ports. In \cite{secondorderstat}, the second-order statistics have been studied where Wong \textit{et al.} derived the ergodic capacity and lower-bounded it. They also derived the level crossing rate and average fade duration of FAS. Their analysis shows a considerable capacity gain resulting from the diversity hidden in a small space of FAS. In \cite{fama}, an extension of FAS for multiple access has been proposed. The analysis of the outage probability and average outage capacity characterizes the multiplexing gain of the fluid antenna multiple access, and shows its capability to support hundreds of users using only one fluid antenna at each user. Further, the problem of port selection was addressed in \cite{portselection}.  \textcolor{black}{In principle, selecting the port with the strongest signal requires SNR observations from all the ports. This may be unfeasible in practice since switching from port to port for SNR observations can result in unbearable delays. Authors in \cite{portselection} use a combination of machine learning methods and analytical approximation on a few observed ports to estimate the strongest signal. Specifically, they show that observing only 10\% of the ports provides more than an order of magnitude reduction in the outage probability.}

The state-of-the-art works demonstrate considerable potential for the arising technology of FAS. Nevertheless, assessing its achievable performance depends entirely on diversity reception over highly correlated channels that follow Jake's model \cite{stuber}. Moreover, while closed-form or single-fold integrals have been derived for the probability of outage of independent channels \cite{goldsmith}, they are usually challenging to obtain for diversity receptions over arbitrarily correlated channels. More specifically, the probability of outage of an \textcolor{black}{\textit{N}}-branch selection combiner in a multi-antenna system is written in terms of the multivariate cumulative distribution function (CDF) of the channel gains.
However, for correlated fading channels, these CDFs have intractable expressions involving nested integrals of the Marcum Q-function. An extensive attention has been dedicated, in the literature, to derive tractable mathematical expressions for the multivariate probability density functions (PDFs) and CDFs of these distributions. One category of approaches considers arbitrary covariance matrices, however, it restricts the number of branches. For instance, simplified CDF representations have been derived for selection combining over correlated channels for bivariate \cite{bivariate1, bivariate2, bivariate3, bivariate4}, trivariate and quadrivariate distributions \cite{trivariate1, trivariate2, trivariate3andquadrivariate1, trivariate4andquadrivariate2}. Another category of approaches considers an arbitrary number of branches but restricts the covariance matrix to specific forms. For instance, the constant correlation model was heavily studied in the literature, and simplified CDF expressions were derived \cite{equal1, equal2}. A more generalized covariance model was considered in \cite{general1, general2, general3}. Nevertheless, the covariance matrix was constrained to a certain form as in \cite[eq. (2)]{general2}. The most general framework to derive the distribution of an \textcolor{black}{\textit{N}}-branch selection combiner over correlated channels provides an \textcolor{black}{\textit{N}} multi-fold integral expression \cite{very_general}, which is not insightful or computationally efficient. A review of selection combining receivers over correlated channels can be found in \cite{survey_corr}. The authors in \cite{FAS1} parameterize the channel coefficients as in \cite{general2}, which simplifies the derivation of the probability of outage. However, although this parameterization is an essential first step towards gaining insight into the technology, it imposes a structure on the covariance model that may not accurately capture the dependence between the FAS ports given by Jake's model.

This paper aims to study the performance of the FAS under a more accurate correlation model that closely follows Jake's model. More specifically, \textcolor{black}{we incorporate more parameters into the channel model from \cite{FAS1} for more flexibility to approximate Jake's correlation model at the cost of increased complexity. However, by carefully choosing our proposed channel model parameters, we derive a tractable expression for the probability of outage that well approximates the simulation results of FAS.} The approximation is made in two stages. In the first stage, the number of multi-fold integrals in the probability of outage is considerably reduced. In the second stage, we investigate a more insightful probability of outage expression by approximating the result of the first stage approximation with a power of a single integral. \textcolor{black}{Numerical results validate our approximations and show a limited performance gain compared to \cite{FAS1}. This is due to our channel model being more complex to capture the detrimental effect of the highly correlated ports on diversity gain.} Therefore, careful modeling and analysis of the FAS channel are required to determine scenarios where this technology can outperform traditional multiple antenna technologies.

The remainder of this paper is organized as follows. In section \rom{2}, we present the FAS channel. In section \rom{3}, our proposed FAS channel model is presented. Our model provides both an exact and an approximated FAS channel representation. However, unlike the exact modeling, the approximated modeling of the FAS channel gives more tractable expression for probability of outage. Section \rom{4} presents our numerical results and section \rom{5} provides some concluding remarks.

\textit{\textbf{Notation}}: We use boldface upper and lower case letters for matrices and column vectors, respectively. $\textcolor{black}{\mathrm{E}}[.]$, $\textcolor{black}{\mathrm{V}}[.]$ and $\textcolor{black}{\mathrm{Cov}}[.,.]$ denote the statistical expectation, variance and covariance respectively. $(.)^{\textcolor{black}{\mathrm{T}}}$ and $|.|$ stand for the transpose and magnitude respectively and $(.)_{m,n}$ denotes the element in row $m$ and column $n$. \textcolor{black}{ We use \#$(.)$ to denote cardinality}. Finally, $\mathcal{CN}$ denotes the circularly-symmetric complex normal distribution.

\section{FAS Channel}
In this paper, we follow the same abstraction of FAS as in \cite{FAS1}.  We consider a single antenna that can move freely along \textcolor{black}{\textit{N}} equally distributed positions (i.e., ports) on  a linear space of length $W\lambda$, where $\lambda$ is the wavelength of the radiation, as it is shown in \textcolor{black}{F}ig. \ref{fig:FAS}. Therefore, taking the first port as a reference point, the distance between the first port and the $k$-th port is given by

\begin{equation}
    \Delta d_{k,1} = \frac{k-1}{N-1} W \lambda, ~ \text{for } k = 1, 2, \ldots, N.  
\end{equation}

The received signal by the $k$-th port can be modeled as
\begin{equation}
    \textcolor{black}{r_k} = g_k \textcolor{black}{q} + n_k, ~ \text{for } k = 1, 2, \ldots, N.
\end{equation}
\noindent where $\textcolor{black}{q}$ is the transmitted data symbol, \textcolor{black}{$n_k$ is a complex additive white Gaussian noise (AWGN) at the $k$-th port, with zero mean and variance $\sigma_n ^2$, and $g_k$ is the flat fading coefficient at the $k$-th port, following a circularly symmetric complex Gaussian
distribution with zero mean and variance $\sigma^2$.}

Furthermore, we assume that the channel coefficients of $\textbf{g} = (g_1, g_2, \ldots, g_N)^{\textcolor{black}{\mathrm{T}}}$ are correlated with a covariance matrix $\textcolor{black}{\boldsymbol{\Sigma}}_g$. Assuming two dimensional isotropic scattering with an isotropic receiver port, as in \cite{FAS1}, the spatial separation between the ports of FAS yields a difference in the phases of arriving paths, thus inducing correlation between the channels following Jake's model \cite{stuber} as 
\begin{equation}
(\boldsymbol{\Sigma}_g)_{k,\ell} = \mathrm{Cov}[g_{k}, g_{\ell}] = \sigma^{2} J_{0}\left(2 \pi \frac{\Delta d_{k, \ell}}{\lambda}\right)= \sigma^{2} J_{0}\left(\frac{2 \pi(k-\ell)}{N-1} W\right).
\end{equation}
\noindent where $J_0(.)$ is the zero-order Bessel function of the first kind. 

Further, the average signal-to-noise ratio (SNR) at each port is given by
\begin{equation}
    \Gamma = \textcolor{black}{ \frac{ \textcolor{black}{\mathrm{E}}[|g_k \textcolor{black}{q}|^2]}{\textcolor{black}{\mathrm{E}}[|n_k|^2]} =  \frac{\sigma^2 \textcolor{black}{\mathrm{E}}[|\textcolor{black}{q}|^2]}{\sigma_n ^2}} = \sigma^2 \Theta, 
\end{equation}

\noindent where $\Theta \delequal \frac{ \textcolor{black}{\mathrm{E}}[|\textcolor{black}{q}|]^2}{\sigma_n ^2}$. Assuming that the FAS can instantly\textcolor{black}{ \footnote{ \textcolor{black}{The instant switching between ports can be hard to achieve in practice since moving physical materials results in delay. However, one direction towards fast port switching consists of considering smaller antenna sizes at higher frequencies \cite{FAS1}. Another possibility consists of incorporating pixel antennas into the FAS design. Specifically, FAS can be based on an array of digitally controlled mini pixels that go on and off with a negligible delay \cite{pixels}.}}} switch to \textcolor{black}{the position of the maximum magnitude of the channel coefficients as in \cite{FAS1}, we are interested in the distribution of the following random variable}

\begin{equation}
    g_{\textcolor{black}{\mathrm{FAS}}} = \max{\{|g_1|, |g_2|,\ldots, |g_N|\}}.
\end{equation}

\begin{figure}[!t]
\centering
\includegraphics[width=\linewidth]{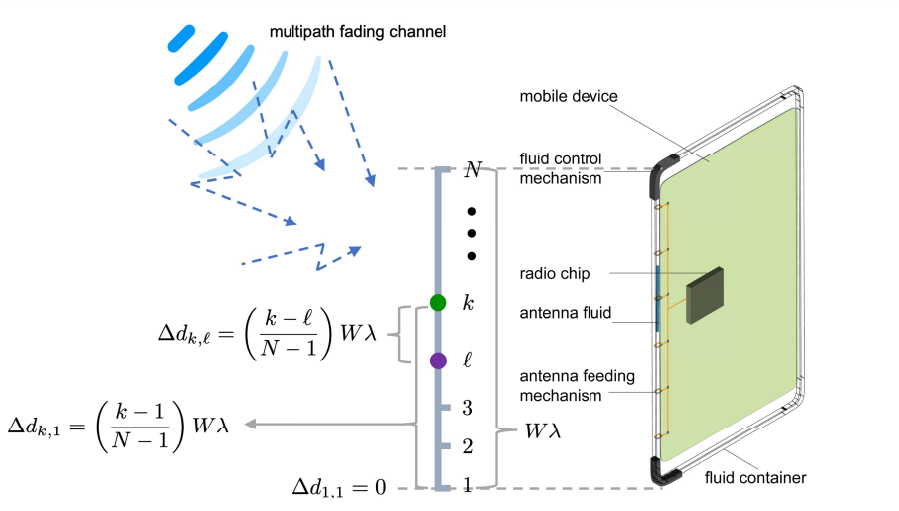}
\caption{ Architecture of FAS \cite{FAS1}.}
\label{fig:FAS}
\end{figure}

\section{Proposed FAS Channel Model}

\subsection{General Model for Arbitrarily Correlated Rayleigh Channels} \label{SectionGeneralModel}

We propose to represent a set of correlated Rayleigh fading channels $h_1, h_2,\ldots, h_N$ as 
\begin{equation}
h_k = \sigma_h \sqrt{1-\sum_{l=1}^{M} \alpha_{k, l}^{2}}~~(x_{k}+\textcolor{black}{\mathrm{j}} y_k)+ \sigma_h \sum_{l=1}^{M} \alpha_{k, l} (a_{l}+\textcolor{black}{\mathrm{j}} b_{l}), 
\label{eq4}
\end{equation}
\noindent where \textcolor{black}{$\mathrm{j}$ is the imaginary unit}, $x_1,\ldots, x_N, y_1,\ldots, y_N, a_1,\ldots, a_M, b_1,\ldots, b_M$ are independent and identically distributed (i.i.d) normal random variables with  zero-mean and variance $\frac{1}{2}$. Furthermore, $\alpha_{1,1},\ldots, \alpha_{1,M},\ldots,$ $\alpha_{N,1},\ldots, \alpha_{N,M}$, $M$ and $\sigma_h $ are parameters to be chosen according to the different correlations between the channels. We can see that according to the model
\begin{align}
    h_k &\sim \mathcal{CN}(0,\,\sigma_{h}^{2})\, ~\forall k \in \{1, \ldots, N\}, \\
    (\boldsymbol{\Sigma}_h)_{\textcolor{black}{m,n}} &= \left\{
    \begin{array}{ll}
        \sigma_h^2 & \mbox{if } \textcolor{black}{m}=\textcolor{black}{n} \\
        \sigma_h^2 \sum_{l=1}^{M} \alpha_{\textcolor{black}{m},l} \alpha_{\textcolor{black}{n},l} & \mbox{if } \textcolor{black}{m} \neq \textcolor{black}{n}
        \label{cov_h}
    \end{array},
\right.
\end{align}
where $\boldsymbol{\Sigma}_h$ is the covariance matrix of $\textbf{h}=(h_1, h_2,\ldots, h_N)^{\textcolor{black}{\mathrm{T}}}$.

We can see that the channel modeling in \cite{FAS1, secondorderstat , fama , portselection} is a particular case of the model in (\ref{eq4}) with $M=1$, $\sigma_h = \sigma$, $\alpha_{1,1} = 1$ and $\alpha_{k,1} = J_{0} \left( \frac{2 \pi (k-1)}{(N-1)} W\right) $ for $k \in \{2, \ldots, N\}$. Therefore, we can see from (\ref{cov_h}) that the correlation matrix of \cite{FAS1} does not exactly follow Jake's model because $\mathrm{Cov} [h_k, h_l] = \sigma^2 J_{0} \left( \frac{2 \pi (k-1)}{(N-1)} W\right) J_{0} \left( \frac{2 \pi (\ell-1)}{(N-1)} W\right)  \neq \sigma^2 J_{0} \left( \frac{2 \pi (k-\ell)}{(N-1)} W\right)$ for $k, \ell \in \{2, \ldots, N\}$ . \textcolor{black}{This problem has been pointed out also in \cite{9830377}.}

\textcolor{black}{The intuition behind our proposed model is to start by a model as in \cite{FAS1, secondorderstat , fama , portselection}, then add more flexibility by incorporating more parameters, at the cost of increased complexity.} In the following, we start by finding our model's parameters for an exact representation of $\textbf{g}=(g_1, g_2,\ldots, g_N)^{\textcolor{black}{\mathrm{T}}}$ as in (\ref{eq4}). Then, motivated by this representation, we choose a different set of model parameters that approximate the joint distribution of $\textbf{g}=(g_1, g_2,\ldots, g_N)^{\textcolor{black}{\mathrm{T}}}$ while maintaining mathematical tractability of the CDF of $g_{\textcolor{black}{\mathrm{FAS}}} = \max{\{|g_1|, |g_2|,\ldots, |g_N|\}}$, under a less-idealized correlation model than the one considered in \cite{FAS1, secondorderstat , fama , portselection}.

\subsection{Exact Model} \label{exact_model} 

\indent In order to represent $\textbf{g}=(g_1, g_2,\ldots, g_N)^{\textcolor{black}{\mathrm{T}}}$ as in (\ref{eq4}), we choose the parameters $\sigma_h$, $M$ and $\alpha_{k,l}$ for $k=1,\ldots, N$ and $l=1, \ldots, M$, such that the random vectors $\textbf{g}=(g_1, g_2,\ldots, g_N)^{\textcolor{black}{\mathrm{T}}}$ and $\textbf{h}=(h_1, h_2,\ldots, h_N)^{\textcolor{black}{\mathrm{T}}}$ have the same joint distribution.

\begin{theorem}

\textcolor{black}{Let $s_1, s_2,\ldots, s_N$ be the non-increasingly ordered eigenvalues of $\boldsymbol{\Sigma}_g$ and let  $\textbf{u}_1, \textbf{u}_2,\ldots, \textbf{u}_N$ be their respective associated eigenvectors with $\textbf{u}_l = (u_{1,l},\ldots, u_{N,l})^{\textcolor{black}{\mathrm{T}}},~ 1 \leq l \leq N $.}

\textcolor{black}{If $\sigma_h = \sigma$, $M=N$ and $\alpha_{k,l} = \frac{\sqrt{s_l}}{\sigma} u_{k,l},~ 1\leq k,l \leq N $, then $\textbf{g}=(g_1, g_2,\ldots, g_N)^{\textcolor{black}{\mathrm{T}}}$ and $\textbf{h}=(h_1, h_2,\ldots, h_N)^{\textcolor{black}{\mathrm{T}}}$ have the same joint distribution. In this case, $\textbf{g}=(g_1, g_2,\ldots, g_N)^{\textcolor{black}{\mathrm{T}}}$ can be represented using model (\ref{eq4}) as 
\begin{equation}
g_k \overset{\operatorname{\mathrm{d}}}{=} \sum_{l=1}^{N} \sqrt{s_l} u_{k,l} (a_{l}+\textcolor{black}{\mathrm{j}}~b_{l}),~ \text{for } k\in \{1,\ldots, N\}, \label{eq:theo12}
\end{equation}
where the operator $\overset{\operatorname{\mathrm{d}}}{=}$ denotes equality in the sense of distribution. }
\label{theo1}
\end{theorem}

\begin{proof}
    \textcolor{black}{See Appendix \ref{appendixA}}.
\end{proof}

We can see that the model parameters choice in \textcolor{black}{T}heorem \ref{theo1} provides an exact representation of the correlated channel vector $\textbf{g}$ using model (\ref{eq4}). Nevertheless, the representation in (\ref{eq:theo12}) is only in terms of the random variables $a_1,\ldots,a_N$, $b_1, \ldots, b_N$ and does not include the random variables $x_1,\ldots,x_N$, $y_1,\ldots, y_N$. This is due to the choice of model parameters making $\sum_{l=1}^{M} \alpha_{k, l}^{2} = 1$, and thus, multiplying the terms $\{x_k + \textcolor{black}{\mathrm{j}} y_k, \forall k \in \{1, \ldots, N\} \} $ by zeros in the model. However, the random variables $x_1,\ldots,x_N$, $y_1,\ldots, y_N$ can be very useful in deriving the CDF of $g_{\textcolor{black}{\mathrm{FAS}}} = \max{\{|g_1|, |g_2|,\ldots, |g_N|\}}$, as it will be shown in the following. In the case of the exact model, the CDF of the selection combiner over $\textbf{g}=(g_1, g_2,\ldots, g_N)^{\textcolor{black}{\mathrm{T}}}$, and equivalently, the CDF of the FAS channel, can only be written in the form of \textcolor{black}{\textit{N}}-fold integrals for $N>3$ \cite{Lfolds}. To further simplify the expression of the CDF, we approximate the joint distribution of $\textbf{g} = (g_1, g_2,\ldots, g_N)^{\textcolor{black}{\mathrm{T}}}$ using the same model (\ref{eq4}) in two stages.

\subsection{Approximated Model}

\subsubsection{First Stage Approximation} \label{sec:1stage_approx}\hfill

\begin{definition}
We define the random vector $\hat{\textbf{g}}=(\hat{g}_1, \hat{g}_2,\ldots, \hat{g}_N)^{\textcolor{black}{\mathrm{T}}}$, the approximation of the channel vector $\textbf{g}=(g_1, g_2,\ldots, g_N)^{\textcolor{black}{\mathrm{T}}}$, as 
\begin{align}
    &\hat{g}_k = \sqrt{ \sigma^2 -\sum_{l=1}^{\epsilon\text{-rank}} s_l u_{k, l}^{2}}~~(x_{k}+\textcolor{black}{\mathrm{j}} y_k)+ \sum_{l=1}^{\epsilon\text{-rank}} \sqrt{s_l} u_{k,l} (a_{l}+\textcolor{black}{\mathrm{j}} b_{l}),~~ \forall k \in \{1,\ldots, N\}, 
    \label{approx_cond}
\end{align}
where $\epsilon$-rank is the number of eigenvalues of $\boldsymbol{\Sigma}_g$ exceeding a threshold $\epsilon >0$, $s_1, s_2,\ldots, s_{\epsilon\text{-rank}}$ are the non-increasingly ordered eigenvalues of $\boldsymbol{\Sigma}_g$ greater than $\epsilon$ and $\textbf{u}_1, \textbf{u}_2,\ldots, \textbf{u}_{\epsilon\text{-rank}}$ are their respective associated eigenvectors with $\textbf{u}_l = (u_{1,l},u_{2,l},\ldots, u_{N,l})^{\textcolor{black}{\mathrm{T}}} \text{ for } l\in \{1, 2,\ldots, \epsilon\text{-rank}\}$.  
\label{def:approx}
\end{definition}

Similarly to the exact model (\textcolor{black}{T}heorem \ref{theo1}), \textcolor{black}{D}efinition \ref{def:approx} of the approximation takes the parameters $\sigma_h = \sigma$ and $\alpha_{k,l} = \frac{\sqrt{s_l}}{\sigma} u_{k,l},~ 1\leq k\leq N,~1\leq l\leq M$. However, unlike the exact representation where $M=N$, we introduce more flexibility to the approximation by taking $M = \epsilon\text{-rank}$ for $\epsilon >0$. Then, we investigate the choice of $\epsilon >0$, and equivalently the choice of $\epsilon\text{-rank}$, that provides an analytically tractable approximation of the CDF of $g_{\textcolor{black}{\mathrm{FAS}}} = \max{\{|g_1|, |g_2|,\ldots, |g_N|\}}$ while ensuring that $\hat{\textbf{g}} = (\hat{g}_1, \hat{g}_2,\ldots, \hat{g}_N)^{\textcolor{black}{\mathrm{T}}}$ is close to $\textbf{g} = (g_1, g_2,\ldots, g_N)^{\textcolor{black}{\mathrm{T}}}$ in the sense of distribution.

\begin{theorem}
  \textcolor{black}{Let $W_2$ be the Fréchet distance between two distributions \cite{FrechetDistance}. It should hold that
  \begin{equation}
      W_2(\mathcal{CN}(\textbf{0}_{N \times 1}, \boldsymbol{\Sigma}_g),\mathcal{CN}(\textbf{0}_{N \times 1}, \boldsymbol{\Sigma}_{\hat{g}})) \leq N \epsilon + (N- \epsilon\text{-rank}) \epsilon^2,
  \end{equation}
  where $\boldsymbol{\Sigma}_{\hat{g}}$ denotes the covariance matrix of $\hat{\textbf{g}} = (\hat{g}_1, \hat{g}_2,\ldots, \hat{g}_N)^{\textcolor{black}{\mathrm{T}}}$.}
  \label{distance}
\end{theorem}

\begin{proof} 
See Appendix \ref{appendixB}
\end{proof}

\begin{theorem}
Consider the random vector $\textbf{g} = (g_1, g_2,\ldots, g_N)^{\textcolor{black}{\mathrm{T}}}$ and its approximation $\hat{\textbf{g}} = (\hat{g}_1, \hat{g}_2,\dots, \hat{g}_N)^{\textcolor{black}{\mathrm{T}}}$ given in \textcolor{black}{D}efinition  \ref{def:approx}.
\begin{align}
    \max{\{|\hat{g}_1|, |\hat{g}_2|,\ldots, |\hat{g}_N|\}} &\xrightarrow{\textcolor{black}{\mathrm{d}}} \max{\{|g_1|, |g_2|,\ldots, |g_N|\}} = g_{\textcolor{black}{\mathrm{FAS}}}, ~\text{as}~ \epsilon \xrightarrow{} 0.
\end{align}
\label{convergence}
\end{theorem}

\begin{proof} 
Let $f: \mathbb{C}^N \xrightarrow{} \mathbb{R}$ such that $ f(z_1,\ldots, z_N) = \max\{|z_1|,\ldots, |z_N|\}$. The function $f$ is continuous, \textcolor{black}{and we can see that $\hat{\textbf{g}} \xrightarrow{\mathrm{d}} \textbf{g}, ~\text{as}~ \epsilon \xrightarrow{} 0$}.  Therefore, by the continuous mapping theorem \cite{wiki}, $f(\hat{\textbf{g}}) \xrightarrow{\textcolor{black}{\mathrm{d}}} f(\textbf{g}), ~\text{as}~ \epsilon \xrightarrow{} 0$. 
\end{proof}

\textcolor{black}{Theorem \ref{distance} provides an upper-bound on the Fréchet distance between the distribution of $\textbf{g}$ and its approximation $\hat{\textbf{g}}$. Therefore, it allows to assess the precision level of the approximation as a function of $\epsilon$ and $\epsilon$-rank. We can see that for a fixed $N$, the smaller $\epsilon$, the higher the precision level of the approximation.} By \textcolor{black}{T}heorem \ref{convergence}, we can approximate $g_{\textcolor{black}{\mathrm{FAS}}} = \max \{ |g_1|,\ldots,|g_N|\}$ by approximating the channel vector $\textbf{g} = (g_1, \ldots, g_N)^{\textcolor{black}{\mathrm{T}}}$. Furthermore, \textcolor{black}{by Theorems \ref{distance} and \ref{convergence}, the approximation improves as $\epsilon$ decreases, which in turn increases $\epsilon\text{-rank}$}. Therefore, the best possible approximation is obtained by maximizing $\epsilon\text{-rank}$ and taking it equal to $N$. However, if $M = \epsilon\text{-rank} = N$ then, by \textcolor{black}{T}heorem \ref{theo1}, we obtain an exact representation of the channel vector $\textbf{g} = (g_1, \ldots, g_N)^{\textcolor{black}{\mathrm{T}}}$ in terms of only the random variables $a_1,\ldots,a_N$ and $b_1, \ldots,b_N$. In this case, the CDF of the FAS channel can only be written in the form of \textcolor{black}{\textit{N}}-fold integrals for $N>3$ \cite{Lfolds}. 

On the other hand, including the random variables $x_1,\ldots,x_N$, $y_1,\ldots, y_N$ in the representation of the approximation can lead to a more simplified expression of the joint CDF of $(|g_1|,\ldots,|g_N|)$, and consequently, the CDF of $g_{\textcolor{black}{\mathrm{FAS}}} = \max{\{|g_1|, |g_2|,\ldots, |g_N|\}}$. In fact, by representing $\textbf{g}$ in terms of $x_1,\ldots,x_N$, $y_1,\ldots, y_N$, the random vector $(|g_1|,\ldots,|g_N|)$ becomes independent conditionally on $a_1,\ldots,a_M$, $b_1, \ldots, b_M$. As a result, the conditional joint CDF of the channel magnitudes becomes the product of the conditional CDFs of each channel magnitude, which constitutes the first step into deriving the CDF of the FAS channel. Hence, we choose $\epsilon$ such that the approximation in \textcolor{black}{D}efinition \ref{def:approx} remains in terms of all the random variables $x_1,\dots, x_N$, $y_1,\dots, y_N$, $a_1,\dots, a_M$ and $b_1,\dots, b_M$. In other words, we choose $\epsilon\text{-rank} < N$.

\begin{theorem} [\textbf{Cumulative Distribution Function}]\hfill\\
    Consider the approximation $\hat{\textbf{g}} = (\hat{g}_1, \hat{g}_2,\dots, \hat{g}_N)^{\textcolor{black}{\mathrm{T}}}$ given in \textcolor{black}{D}efinition \ref{def:approx} for $\epsilon > 0$ such that $\epsilon\text{-rank} < N$. The CDF of $\max \{|\hat{g}_1|,\ldots,|\hat{g}_N|\} $ is given by
    \begin{align}
    \begin{split}
        &F_{\max \{|\hat{g}_1|,\ldots,|\hat{g}_N|\}}(r)\\ &=  \textcolor{black}{\idotsint\limits_{-\infty}^{\infty}}  \prod_{l=1}^{\epsilon\text{-rank}} \frac{1}{\pi} \exp{(-(a^2_l+b^2_l))}\\ 
        & \prod_{k=1}^{N} \left( 1-Q_1 \left(\frac{\sqrt{2\left(\sum \limits_{l=1}^{\epsilon\text{-rank}} \sqrt{s_l} u_{k,l} a_l\right)^2+ 2\left(\sum \limits_{l=1}^{\epsilon\text{-rank}} \sqrt{s_l} u_{k,l} b_l\right)^2}}{\sqrt{\sigma^2 -\sum \limits_{l=1}^{\epsilon\text{-rank}} s_l u_{k, l}^{2}} }, \frac{\sqrt{2} r}{ \sqrt{\sigma^2 -\sum \limits_{l=1}^{\epsilon\text{-rank}} s_l u_{k, l}^{2}} }\right) \right) \\ &da_1\ldots da_{\epsilon\text{-rank}} ~db_1\ldots db_{\epsilon\text{-rank}},
        \end{split}
        \label{cdfgfas}
    \end{align}
    \label{CDF_g_hat}
\end{theorem}
\noindent \textcolor{black}{where $Q_1(\cdot)$ is Marcum Q-function.}

\begin{proof}
\textcolor{black}{See Appendix \ref{appendixC}}.
\end{proof}

\noindent \textcolor{black}{\textbf{Probability of Outage:} We define the outage event as
\begin{equation}
    \left\{g_{\textcolor{black}{\mathrm{FAS}}}^{2} \Theta<\gamma_{\mathrm{th}}\right\}=\left\{g_{\textcolor{black}{\mathrm{FAS}}}<\sqrt{\frac{\gamma_{\mathrm{th}}}{\Theta}}\right\} = \left\{ \max\{ |g_1|,\ldots, |g_N|\}<\sqrt{\frac{\gamma_{\mathrm{th}}}{\Theta}}\right\}.
\end{equation}
Therefore, the probability of outage can be approximated by 
\begin{equation}
    P_{out}(\gamma_{\mathrm{th}}) \approx F_{\max \{|\hat{g}_1|,\ldots,|\hat{g}_N|\}}(\sqrt{\frac{\gamma_{\mathrm{th}}}{\Theta}}).
    \label{p_out1}
\end{equation}}

The results in (\ref{cdfgfas}) and (\ref{p_out1}) provide approximated expressions for the CDF and probability of outage of the FAS channel that depend on the parameter $\epsilon$. In what follows, we investigate the impact of this parameter on both the accuracy and the mathematical tractability of these expressions. 

\begin{prop}
Let $\epsilon > 0$ be big enough such that $\epsilon\text{-rank} < \frac{N}{2}$. Then, the number of the multi-fold integrals in the approximated CDF and outage probability of the FAS channel is reduced by $\frac{N-2~\epsilon\text{-rank}}{N}$.
\label{gain}
\end{prop}
\begin{proof} \textcolor{black}{Originally, the CDF of the selection combiner over correlated channels is represented by \textit{N}-fold integrals  \cite{N_multifold}. Therefore, the proof follows directly from the fact that (\ref{cdfgfas}) and (\ref{p_out1}) have $2 \times \epsilon\text{-rank}$ multi-fold integrals.}
\end{proof}

Proposition \ref{gain}  shows a reduction in the number of multi-fold integrals as $\epsilon$ increases (i.e. $\epsilon$-rank decreases), while \textcolor{black}{T}heorems \textcolor{black}{\ref{distance}} and \ref{convergence} show a higher accuracy of the approximation as $\epsilon$ decreases (i.e. $\epsilon$-rank increases). Therefore, we have to undergo the trade-off between tractability and accuracy. In practice (as the numerical results illustrate in \textcolor{black}{S}ection \ref{sec:results}), we can find an $\epsilon$ \textcolor{black}{small} enough to have a considerable reduction in the number of multi-fold integrals while guaranteeing a high accuracy of the approximation. However, we can see that finding a suitable $\epsilon$ is not enough to evaluate the expressions of the approximations. The approximated probability of outage and CDF do not depend explicitly on $\epsilon$, but instead on $\epsilon$-rank. Therefore, we need to count the number of eigenvalues of $\boldsymbol{\Sigma}_g$ exceeding $\epsilon$, which numerically is a straightforward task. However, determining $\epsilon$-rank as a function of the problem parameters (\textit{i.e.} $N$ and $W$) is more insightful and provides stand-alone expressions of the FAS channel distribution and probability of outage. In what follows, we investigate the eigenvalue distribution function of the covariance matrix $\boldsymbol{\Sigma}_g$ to approximate the number of its eigenvalues exceeding a certain threshold. 

\begin{theorem}
Consider the matrix $\textbf{T}_N$ of size $N \times N$ such that 
$$
(\textbf{T}_N)_{(k,l)} = J_0 \left( 2 \pi (k-\ell) c \right), \text{ for } k,\ell \in \{1,\ldots,N\} \text{ and } 0 < c < \frac{1}{2}.
$$ 
Let $\{ s_{N,k}; k=1,\ldots,N\}$ be the set of its eigenvalues. Consider the eigenvalue distribution function of $\textbf{T}_N$ defined as $D_{N}(x)=\left(\text { number of } s_{N, k} \leq x\right) / \mathrm{N}$, and its limiting distribution $D(x)=\lim _{N \rightarrow \infty} D_{N}(x)$. It should hold that

\begin{equation}
    D(x) = \left \{ 
    \begin{array}{ll}
         1 - 2c & \text{ if } 0 < x < \frac{\sigma^2}{\pi c} \\
         1 - 2c + \sqrt{ (2 c)^2 - \frac{4 \sigma^4}{(\pi x)^2}} & \text{ if } x \geq \frac{\sigma^2}{\pi c},
    \end{array}
    \right.
\end{equation}
\label{theo:epsilonrank}
\end{theorem}

\begin{proof}
    See Appendix \ref{AppendixD}.
\end{proof} 

We can see from \textcolor{black}{T}heorem \ref{theo:epsilonrank} that for a small enough threshold $x$, the fraction of eigenvalues of $\textbf{T}_N$ less than $x$ (as $N$ goes to infinity) become independent of $x$, and always equal to $1-2c$. This result allows us to approximate the $\epsilon$-rank of the covariance matrix $\boldsymbol{\Sigma}_g$. In fact, we can see that for $c = \frac{W}{N-1}$, $\boldsymbol{\Sigma}_g = \textbf{T}_N$. Therefore, for a small $ 0 < \epsilon < \frac{\sigma^2}{\pi c}$ and a large $N$, the fraction of eigenvalues of $\boldsymbol{\Sigma}_g$ less than $\epsilon$ can be approximated by $1-\frac{2W}{N-1}$. Therefore, the fraction of eigenvalues exceeding $\epsilon$ can be approximated by $\frac{2W}{N-1}$. Thus,
\begin{equation}
    \epsilon\text{-rank} \approx 2W \frac{N}{N-1}, \text{ for large } N \text{ and } 0 < \epsilon < \frac{\sigma^2 (N-1)}{\pi W}.
    \label{epsilon-rank-approx0}
\end{equation}

\noindent The approximation above is more accurate as $N$ goes to infinity. We observe through simulation that the asymptotic convergence is slow, and the approximation of $\epsilon$-rank becomes accurate for very large $N$ beyond the range we test for FAS. Therefore, we propose to approximate $\epsilon$-rank as follows
\begin{equation}
    \epsilon\text{-rank} \approx a ~W \frac{N}{N-1},\text{ for large } N \text{ and } 0 < \epsilon < \frac{\sigma^2 (N-1)}{\pi W}. 
    \label{epsilon-rank-approx}
\end{equation}
\noindent where $a$ is a constant that we determine in the numerical results section to approximate $\epsilon\text{-rank}$ better for the tested range of $N$. 

Even though the practical values of $\epsilon$-rank are quite small compared to $\frac{N}{2}$, multi-fold integral representation of the CDF and probability of outage still restricts us from gaining insights about the FAS. Therefore, we design a second stage approximation of the channel vector to further simplify the CDF expression. \\

\subsubsection{Second Stage Approximation} \hfill \\
In the first stage approximation, we use the random vector $\hat{\textbf{g}} = (\hat{g}_1, \ldots, \hat{g}_N)^{\textcolor{black}{\mathrm{T}}}$ to approximate the distribution of the channel vector $\textbf{g} = (g_1, \ldots, g_N)^{\textcolor{black}{\mathrm{T}}}$. Then, we approximate the FAS channel distribution $g_{\textcolor{black}{\mathrm{FAS}}} = \max \{ |g_1|, \ldots, |g_N| \}$ by the distribution of $\max \{ |\hat{g}_1|, \ldots, |\hat{g}_N| \}$. The second stage approximation aims to further approximate the distribution of $\max \{ |\hat{g}_1|, \ldots, |\hat{g}_N| \}$. However, it starts by approximating the joint distribution of a set of independent \textcolor{black}{random vectors, each having the same distribution as} $\hat{\textbf{g}} = (\hat{g}_1, \hat{g}_2,\ldots, \hat{g}_N)^{\textcolor{black}{\mathrm{T}}}$, then uses it to retrieve an approximation of the CDF of  $\max \{ |\hat{g}_1|, \ldots, |\hat{g}_N| \}$.

\textcolor{black}{In more detail, we consider a random matrix $\hat{\textbf{G}}$ with independent columns, each having the same distribution as $\hat{\textbf{g}} = (\hat{g}_1, \hat{g}_2,\ldots, \hat{g}_N)^{\textcolor{black}{\mathrm{T}}}$. Then, we design a random matrix $\Tilde{\textbf{G}}$ such that its joint distribution approximates the joint distribution of $\hat{\textbf{G}}$.} Finally, by exploiting how $\hat{\textbf{G}}$ and $\hat{\textbf{g}}$ are related, we obtain an approximation of the distribution of $\max \{ |\hat{g}_1|, \ldots, |\hat{g}_N|\}$ from the distribution of $\Tilde{\textbf{G}}$ that we carefully design to avoid a multi-fold integral representation of the CDF expression.

\begin{definition}
    \textcolor{black}{Let $R \in \{1,\ldots, N\}$, $\epsilon >0$ such that $\epsilon\text{-rank} < N$. We define the random matrices $\hat{\textbf{G}}$ and $\Tilde{\textbf{G}}$ of size $N \times R$ such that their $(k,r)$ entries are respectively written as 
    \begin{align}
        \hat{g}_{k,r} &= \sqrt{ \sigma^2 -\sum_{l=1}^{\epsilon\text{-rank}} s_l u_{k, l}^{2}}~~(x_{k,r}+\textcolor{black}{\mathrm{j}} y_{k,r})+ \sum_{l=1}^{\epsilon\text{-rank}} \sqrt{s_l} u_{k,l} (a_{l,r}+\textcolor{black}{\mathrm{j}} b_{l,r}), \label{def:g_hat} \\
        \Tilde{g}_{k,r} &= \sqrt{ \sigma^2 -\sum_{l=1}^{\epsilon\text{-rank}} s_l u_{k, l}^{2}}~~(x_{k,r}+ \textcolor{black}{\mathrm{j}} y_{k,r})+ \sum_{l=1}^{\epsilon\text{-rank}} \sqrt{s_l} u_{k,l} (a_{l,k}+\textcolor{black}{\mathrm{j}} b_{l,k}), \label{def:g_tilde}
    \end{align}
    \noindent where $\{x_{k,r}, y_{k,r}, a_{l,p}, b_{l,p}, 1 \leq r \leq R, 1 \leq k,p \leq N ,1 \leq l \leq \epsilon\text{-rank} \}$ is a set of i.i.d normal random variables with zero-mean and variance $\frac{1}{2}$.}
\end{definition}

\begin{remark}
 \textcolor{black}{We can see that the random matrix $\hat{\textbf{G}}$ has dependent rows and independent columns, with each column having the same distribution as $\hat{\textbf{g}} = (\hat{g}_1,\ldots, \hat{g}_N)^{\textcolor{black}{\mathrm{T}}}$. On the other hand, $\Tilde{\textbf{G}}$ has independent rows and dependent columns, with each column having the same distribution as $\Tilde{\textbf{g}} = (\Tilde{g}_1, \ldots, \Tilde{g}_N)^{\textcolor{black}{\mathrm{T}}}$ such that
 \begin{equation}
    \Tilde{g}_{k} = \sqrt{ \sigma^2 -\sum_{l=1}^{\epsilon\text{-rank}} s_l u_{k, l}^{2}}~~(x_{k}+\textcolor{black}{\mathrm{j}} y_{k})+ \sum_{l=1}^{\epsilon\text{-rank}} \sqrt{s_l} u_{k,l} (a_{l,k}+\textcolor{black}{\mathrm{j}} b_{l,k}), ~ 1 \leq k \leq N.
 \end{equation}}
\label{indep}
\end{remark}

\begin{prop} 
    \textcolor{black}{If we arrange the elements of $\hat{\textbf{G}}$ and $\Tilde{\textbf{G}}$ as $(\hat{g}_{1,1},\dots, \hat{g}_{N,1},\ldots, \hat{g}_{1,R},\dots, \hat{g}_{N,R})^{\textcolor{black}{\mathrm{T}}}$ and $(\Tilde{g}_{1,1},\dots, \Tilde{g}_{1,R},\ldots, \Tilde{g}_{N,1},\dots, \Tilde{g}_{N,R})^{\textcolor{black}{\mathrm{T}}}$ respectively, then their mean vectors $\boldsymbol{\mu}_{\hat{G}}$ and $\boldsymbol{\mu}_{\Tilde{G}}$ and covariance matrices $\boldsymbol{\Sigma}_{\hat{G}}$ and $\boldsymbol{\Sigma}_{\Tilde{G}}$ can be respectively written as
    \begin{align}
        &\boldsymbol{\mu}_{\hat{G}}    = (0, 0, \ldots, 0)^{\textcolor{black}{\mathrm{T}}}, &&\boldsymbol{\mu}_{\Tilde{G}} = (0, 0, \ldots, 0)^{\textcolor{black}{\mathrm{T}}}, \\
        &\boldsymbol{\Sigma}_{\hat{G}} = 
        \begin{pmatrix}
        \boldsymbol{\Sigma}_{\hat{g}} & 0 & \cdots & 0 \\
        0 & \boldsymbol{\Sigma}_{\hat{g}} & \cdots & 0 \\
        \vdots  & \vdots  & \ddots & \vdots  \\
        0 & 0 & \cdots & \boldsymbol{\Sigma}_{\hat{g}}
    \end{pmatrix}, && \boldsymbol{\Sigma}_{\Tilde{G}}= 
    \begin{pmatrix}
    \boldsymbol{\Sigma}_{1} & 0 & \cdots & 0 \\
    0 & \boldsymbol{\Sigma}_{2} & \cdots & 0 \\
    \vdots  & \vdots  & \ddots & \vdots  \\
    0 & 0 & \cdots & \boldsymbol{\Sigma}_{N}
    \end{pmatrix}.
    \end{align}
    where $\boldsymbol{\Sigma}_{\hat{g}}$ is the covariance matrix of $\hat{\textbf{g}}$, and $\boldsymbol{\Sigma}_{k}$ for $1 \leq k \leq N$ is the $R \times R$ matrix defined as
\begin{equation}
    (\boldsymbol{\Sigma}_{k})_{m,n} = \left\{
    \begin{array}{ll}
        \sigma^2 & \mbox{if } m=n \\
        \sum_{l=1}^{\epsilon\text{-rank}} s_l u_{k,l}^2 & \mbox{if } m \neq n
    \end{array}
\right. \text{for } m,n \in \{1,\ldots,R\}.
\end{equation}}
    \label{prop:covG_hat_tilde}
    \end{prop}
\begin{proof}
See Appendix \ref{appendixE}.
\end{proof}

\textcolor{black}{Now that $\hat{\textbf{G}}$ and $\tilde{\textbf{G}}$ are defined and their distributions are determined (\textit{i.e.} mean vectors and covariance matrices), we investigate the distributions of the maximum magnitude of their elements and their relationship with the distribution of $\max \{ |\hat{g}_1|, \ldots |\hat{g}_N|\}$ in the next theorem.}

\begin{theorem}
\textcolor{black}{Let $\Omega_R$ and $\Psi_R$ be the maximum magnitudes of the elements of $\hat{\textbf{G}}$ and $\Tilde{\textbf{G}}$ respectively defined as
\begin{align}
    \Omega_R &= \max\{ |\hat{g}_{k,r}|,~ 1\leq k \leq N, ~ ~ 1\leq r \leq R \},\\
    \Psi_R &= \max\{ |\Tilde{g}_{k,r}|,~ 1\leq k \leq N, ~ ~ 1\leq r \leq R \}.
\end{align}
Then, the CDFs of $\Omega_R$ and $\Psi_R$ can respectively be written as 
\begin{align}
     &F_{\Omega_R}(g) = \left( F_{\max \{|\hat{g}_1|, \ldots,|\hat{g}_N|\} }(g) \right)^R, \label{theo:omega}\\
     \begin{split}
     &F_{\Psi_R}(g) = \\
     &\prod_{k=1}^N \int_{0}^{+\infty} \frac{1}{\sum\limits_{l=1}^{\epsilon\text{-rank}} s_l u_{k, l}^{2}} \exp \left({-\frac{r}{\sum\limits_{l=1}^{\epsilon\text{-rank}} s_l u_{k, l}^{2}}}\right)
         \left( 1-Q_1 \left(\frac{\sqrt{2r}}{\sqrt{\sigma^2 -\sum \limits_{l=1}^{\epsilon\text{-rank}} s_l u_{k, l}^{2}} }, \frac{\sqrt{2}g}{\sqrt{\sigma^2 -\sum \limits_{l=1}^{\epsilon\text{-rank}} s_l u_{k, l}^{2}} }\right) \right)^R dr.
     \end{split} \label{theo:Psi}
\end{align}}
\end{theorem}

\begin{proof}
\textcolor{black}{See Appendix \ref{appendixF}}.
\end{proof}

We design the random matrix $\Tilde{\textbf{G}}$ such that $F_{\max\{ |\Tilde{g}_{k,r}|,~ 1\leq k \leq N,~ 1\leq r \leq R \}}$ is a product of single integrals, as it is shown in (\ref{theo:Psi}). Now, if we choose $R$ such that the distribution of $\Tilde{\textbf{G}}$ approximates the distribution of $\hat{\textbf{G}}$ in a certain sense, then the distribution of $\max\{ |\Tilde{g}_{k,r}|,~ 1\leq k \leq N,~ 1\leq r \leq R \}$ will also approximate the distribution of $\max\{ |\hat{g}_{k,r}|,~ 1\leq k \leq N,~ 1\leq r \leq R \}$, resulting in a single-integral approximation of the CDF of $\max\{ |\hat{g}_{k,r}|,~ 1\leq k \leq N,~ 1\leq r \leq R \}$. On the other hand, (\ref{theo:omega}) shows that $F_{\max \{|\hat{g}_1|, \ldots,|\hat{g}_N|\}} = F^{\frac{1}{R}}_{\max\{ |\hat{g}_{k,r}|,~ 1\leq k \leq N,~ 1\leq r \leq R \}}$. Therefore, the latter approximation allows us to write CDF of $\max \{|\hat{g}_1|, \ldots,|\hat{g}_N|\}$ as a power of single integrals instead of the initial multi-fold integrals. In the following, we investigate the choice of $R$ that allows to approximate $\hat{\textbf{G}}$ by $\Tilde{\textbf{G}}$ in the sense of a distance that we define.

\begin{definition}
    Let $\mathcal{CN}_0$ denote the family of multivariate circularly-symmetric complex normal distributions with zero-mean. We define a distance between two distributions $F$ and $G$ in $\mathcal{CN}_0$ with respective covariance matrices $\boldsymbol{\Sigma}_F$ and $\boldsymbol{\Sigma}_G$ as
    \begin{equation}
        d(F,G) = \norm{\boldsymbol{\Sigma}_F - \boldsymbol{\Sigma}_G}_1, \label{eq:distance_func}
    \end{equation}
   where $\norm{.}_1$ is the induced matrix norm 1 ( i.e. the maximum absolute column sum of the matrix).
\end{definition}

In the following, we denote the covariance matrices of $\hat{\textbf{G}}$ and $\Tilde{\textbf{G}}$ given in \textcolor{black}{P}ropositions \ref{prop:covG_hat_tilde}, by $\boldsymbol{\Sigma}_{\hat{G}}(R)$ and $\boldsymbol{\Sigma}_{\Tilde{G}}(R)$ to emphasize their dependence on $R$, and we use the well-defined distance measure above to approximate $\hat{\textbf{G}}$ by $\Tilde{\textbf{G}}$.

\begin{definition}
    Let $\epsilon >0$ such that $\epsilon\text{-rank} < N$. Consider the random matrices $\hat{\textbf{G}}$ and $\Tilde{\textbf{G}}$ with covariance matrices $\boldsymbol{\Sigma}_{\hat{G}}(R)$ and $\boldsymbol{\Sigma}_{\Tilde{G}}(R)$, respectively. We define $R_1^*$ as the solution to the following optimization problem
    \begin{equation}
        \operatornamewithlimits{min}\limits_{R \in \{1,\ldots, N\} } \norm{\boldsymbol{\Sigma}_{\hat{G}}(R) - \boldsymbol{\Sigma}_{\Tilde{G}}(R)}_1 
        \tag{P1}
        \label{P1}
    \end{equation}
    \end{definition}

We can see that $R_1^*$ minimizes the distance in (\ref{eq:distance_func})  between $\hat{\textbf{G}}$ and $\Tilde{\textbf{G}}$. Therefore, it makes sense to assume that the distribution of $\Tilde{\textbf{G}}$ approximates the distribution of $\hat{\textbf{G}}$ for $R=R_1^*$. However, solving the above optimization problem is not straight forward. Therefore, we relax (\ref{P1}) by (\ref{P2}) defined in the following, where the feasible set $\{1,\ldots,N\}$ becomes restricted to the set of divisors of $N$. 

\begin{definition}
   Let $\epsilon >0$ such that $\epsilon\text{-rank} < N$. Consider the random matrices $\hat{\textbf{G}}$ and $\Tilde{\textbf{G}}$ with covariance matrices $\boldsymbol{\Sigma}_{\hat{G}}(R)$ and $\boldsymbol{\Sigma}_{\Tilde{G}}(R)$, respectively. We define $R_2^*$ as the solution to the following relaxed optimization problem
    \begin{equation}
        \operatornamewithlimits{min}\limits_{R \text{ divisor of } N } \norm{\boldsymbol{\Sigma}_{\hat{G}}(R) - \boldsymbol{\Sigma}_{\Tilde{G}}(R)}_1 
        \tag{P2}
        \label{P2}
    \end{equation}
    \end{definition}
    
Relaxing the optimization problem (\ref{P1}) simplifies the feasible set. Nevertheless, the objective function is still not straight forward to minimize even on the relaxed feasible set. Therefore, we simplify the objective function in an approximated optimization problem (\ref{P3}) defined in the following.

\begin{definition}
    Consider the covariance matrix $\boldsymbol{\Sigma}_g$ of the channel vector $\textbf{g} = (g_1,\ldots,g_N)^{\textcolor{black}{\mathrm{T}}}$. We define $R_3^*$ as the solution to the following optimization problem 
    \begin{equation}
        \operatornamewithlimits{min}\limits_{R \text{ divisor } N } \norm{\boldsymbol{\Sigma}_{G}(R) - \sigma^2 \boldsymbol{\mathbb{I}}(R)}_1
        \tag{P3}
        \label{P3}
    \end{equation}
    where $\boldsymbol{\Sigma}_{G}(R)$ and $\boldsymbol{\mathbb{I}}(R)$ are the $N R \times N R$ matrices defined as
    \begin{align}
        \boldsymbol{\Sigma}_{G} (R) &= 
        \begin{pmatrix}
        \boldsymbol{\Sigma}_{{g}} & 0 & \cdots & 0 \\
        0 & \boldsymbol{\Sigma}_{{g}} & \cdots & 0 \\
        \vdots  & \vdots  & \ddots & \vdots  \\
        0 & 0 & \cdots & \boldsymbol{\Sigma}_{{g}}
        \end{pmatrix} & 
        \boldsymbol{\mathbb{I}}(R) = 
        \begin{pmatrix}
        \boldsymbol{\mathbf{1}}_{R \times R} & 0 & \cdots & 0 \\
        0 & \boldsymbol{\mathbf{1}}_{R \times R} & \cdots & 0 \\
        \vdots  & \vdots  & \ddots & \vdots  \\
        0 & 0 & \cdots & \boldsymbol{\mathbf{1}}_{R \times R}
    \end{pmatrix}
    \end{align}
    with $\boldsymbol{\mathbf{1}}_{R \times R}$ is the $R \times R$ constant matrix with all entries equal to 1.
\end{definition}

Next, we investigate the accuracy of approximating the objective function in (\ref{P1}) by the objective function in (\ref{P3}) in the following theorem. 
\begin{theorem}
    Let $\epsilon >0$ such that $\epsilon\text{-rank} < N$. Consider the random matrices $\hat{\textbf{G}}$ and $\Tilde{\textbf{G}}$ with covariance matrices $\boldsymbol{\Sigma}_{\hat{G}}(R)$ and $\boldsymbol{\Sigma}_{\Tilde{G}}(R)$, respectively.
    For $\epsilon = \frac{\epsilon '}{2N}$, $\epsilon' >0$, we have
    \begin{equation}
        \left| \norm{\boldsymbol{\Sigma}_{\hat{G}}(R) - \boldsymbol{\Sigma}_{\Tilde{G}}(R)}_1 - \norm{\boldsymbol{\Sigma}_{G}(R) - \sigma^2 \boldsymbol{\mathbb{I}}(R)}_1  \right| < \epsilon' .
    \end{equation}
\end{theorem}

\begin{proof}
    See Appendix \ref{appendixG}.
\end{proof}

We can see that taking $\epsilon = \frac{\epsilon '}{2N}$ for a small enough $\epsilon' >0$ ensures that the objective functions in (\ref{P2}) and (\ref{P3}) are very comparable. This justifies approximating the solution of (\ref{P2}) by the solution of (\ref{P3}).

\begin{theorem}
The solution to the optimization problem (\ref{P3}) is given by  $R_3^*$, the greatest divisor of $N$, verifying
\begin{equation}
    J_{0}\left(\frac{2 \pi(R_3^*-1)}{N-1} W\right) \leq 0.5.
    \label{eq47}
\end{equation}
\label{theo11}
\end{theorem}

\begin{proof}
    See Appendix \ref{appendixH}.
\end{proof}
Theorem \ref{theo11} shows that the optimal solution of (\ref{P3}) is the greatest divisor of $N$ verifying (\ref{eq47}). Since, $ J_0(1.52) \approx 0.5$, $R_3^*$ is the greatest divisor of $N$ verifying
\begin{equation}
  R_3^* \leq \left \lfloor\frac{1.52(N-1)}{2\pi W} \right \rfloor.
\end{equation}
\noindent We propose to approximate $\hat{\boldsymbol{G}}$ by $\Tilde{\boldsymbol{G}}$ in the sense of the distance in (\ref{eq:distance_func}) for $R =  R^* = \min \left \{ \left \lfloor\frac{1.52(N-1)}{2\pi W} \right \rfloor, N \right \}$. This allows us to approximate the distribution of $\max\{ |\hat{g}_{k,r}|,~ 1\leq k \leq N,~ 1\leq r \leq R^* \}$ by the distribution of $\max\{ |\Tilde{g}_{k,r}|,~ 1\leq k \leq N,~ 1\leq r \leq R^* \}$. Although, $R^*$ is a solution of (\ref{P3}) only when $\left \lfloor\frac{1.52(N-1)}{2\pi W} \right \rfloor$ is a divisor of $N$, simulation results show that it accurately approximates the distribution of $\max\{ |\hat{g}_{k,r}|,~ 1\leq k \leq N,~ 1\leq r \leq R^* \}$ by the distribution of $\max\{ |\Tilde{g}_{k,r}|,~ 1\leq k \leq N,~ 1\leq r \leq R^* \}$ without necessarily being a divisor of $N$.

\begin{theorem} [\textbf{Cumulative Distribution Function}] \hfill\\
  Let $\epsilon = \frac{\epsilon '}{2N} < \frac{\sigma^2 (N-1) }{\pi W}$ for $\epsilon' >0$. Then, for $\epsilon\text{-rank} = \ceil{a W \frac{N}{N-1}}$ and $R = \min \left \{ \left \lfloor\frac{1.52(N-1)}{2\pi W} \right \rfloor, N \right \}$, the CDF of $ max\{|\hat{g}_1|, |\hat{g}_2|,\ldots, |\hat{g}_N| \}$ can be approximated as
\begin{align}
    \begin{split}
        \textcolor{black}{F_{\max\{|\hat{g}_1|, |\hat{g}_2|,\ldots, |\hat{g}_N|\}}(g) \approx \left( F_{\Psi_R}(g) \right) ^{\frac{1}{R}}.}
    \end{split}
    \label{joint_distribution_g_approx}
    \end{align}
    \label{theo:joint_distribution_g_approx}
\end{theorem}

\begin{proof}
We have $\epsilon < \frac{\sigma^2 (N-1) }{\pi W}$, then, according to (\ref{epsilon-rank-approx}), $\epsilon\text{-rank}$ can be approximated as $a W \frac{N}{N-1}$. Further, we have $\epsilon = \frac{\epsilon '}{2N}$ for $\epsilon' >0$, then, $ F_{\Omega_R} (g) \approx F_{\Psi_R} (g)$ for $R = \min \left \{ \left \lfloor\frac{1.52(N-1)}{2\pi W} \right \rfloor, N \right \}$. On the other hand, we have $ F_{\Omega_R}(g) = \left (F_{\max\{|\hat{g}_1|, |\hat{g}_2|,\ldots, |\hat{g}_N|\}}(g) \right)^R$ according to (\ref{theo:omega}). Therefore, $F_{\max\{|\hat{g}_1|, |\hat{g}_2|,\ldots, |\hat{g}_N|\}}(g)  = \left( F_{\Omega_R}(g)\right)^{\frac{1}{R}} \approx \left( F_{\Psi_R}(g) \right) ^{\frac{1}{R}} $, and the expression in (\ref{joint_distribution_g_approx}) is obtained by plugging in the expression of $F_{\Psi_R}$ given in (\ref{theo:Psi}).
\end{proof}

\noindent \textcolor{black}{\textbf{Probability of Outage:} By approximating the CDF of $\max\{|\hat{g}_1|, |\hat{g}_2|,\ldots, |\hat{g}_N|\}$, which in turn approximates the CDF of $g_{\textcolor{black}{\mathrm{FAS}}} = \max \{ |g_1|, \ldots, |g_N|\}$, We obtain an approximation for the probability of outage of the FAS channel as
\begin{equation}
    P_{out}(\gamma_{\mathrm{th}}) \approx \left( F_{\Psi_R}(\sqrt{\frac{\gamma_{\mathrm{th}}}{\Theta}}) \right)^\frac{1}{R}.
    \label{p_out2}
\end{equation}} 

The second stage approximation allows us to further simplify the probability of outage and the CDF of the FAS. In fact, instead of $N$ multi-fold integrals, we represent the probability of outage and the CDF as a power of single integrals.

\section{Numerical Results}\label{sec:results}

In this section, we start by \textcolor{black}{experimentally motivating the need for a more flexible model that closely approximates Jake's model. We also motivate the design of our proposed model. Then, we assess the accuracy of the first and second stage approximations. Finally, we analyze the FAS performance.}

\paragraph{\textcolor{black}{Motivation}}

\textcolor{black}{ In Fig. \ref{Cov_model_comp}, we illustrate the difference between the correlation model in \cite{FAS1} compared with Jake's model in terms of eigenvalues distribution. Specifically, we plot the fraction of eigenvalues exceeding $\epsilon$ versus $\epsilon$, considering $\sigma=1$, $N=100$, and multiple values of $W$. We can see that the two compared correlation models exhibit considerably different eigenvalue profiles. For instance, for $W= 0.2$, Jake's correlation model has only $3\%$ of eigenvalues exceeding $10^{-4}$, while $97 \%$ of the eigenvalues of the correlation model in Ref. \cite{FAS1} are greater than $10^{-4}$. The apparent trend in Jake's correlation model indicates that a large percentage of the eigenvalues are negligible compared to a few dominant eigenvalues. On the other hand, the correlation model of Ref. \cite{FAS1} seems to have comparable eigenvalues that proportionally contribute to the covariance matrix. Therefore, a more flexible model that closely approximates Jake's model is needed to accurately model the FAS channel.}
\begin{figure}[t]
    \centering
    \includegraphics[width=0.8 \textwidth]{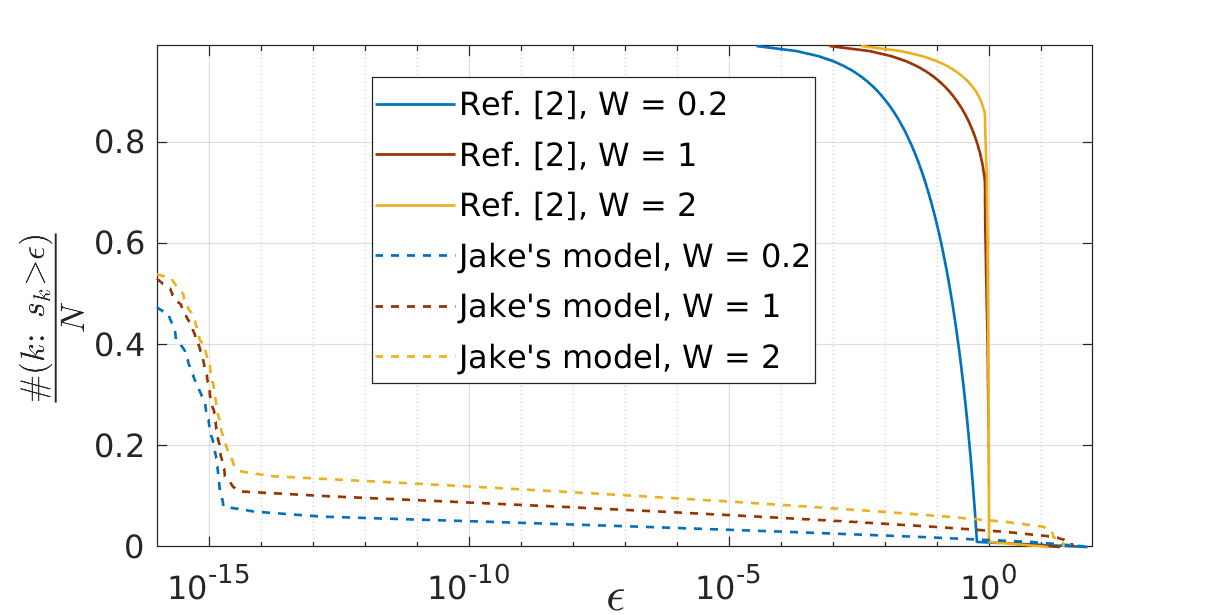}
    \caption{
    \textcolor{black}{Eigenvalues distribution of Jake's correlation model compared with Ref. \cite{FAS1} correlation model for $\sigma=1$, $N=100$, and different values of $W$.} }
    \label{Cov_model_comp}
\end{figure}

 The main intuition behind the way we design the \textcolor{black}{first stage} approximation comes from the decaying profile of the eigenvalues of the covariance matrix $\boldsymbol{\Sigma}_g$ \textcolor{black}{corresponding to Jake's model}. In fact, the covariance matrix $\boldsymbol{\Sigma}_g$ differs from the covariance matrix of its approximation, $\boldsymbol{\Sigma}_{ \hat{g}}$, only at the off-diagonal elements. More precisely, we can write the $(\textcolor{black}{m,n})$-th off-diagonal entry of $\boldsymbol{\Sigma}_g$ as $\sum_{l=1}^N s_{\textcolor{black}{m}} u_{\textcolor{black}{m},l} u_{\textcolor{black}{n},l} $, while the $(\textcolor{black}{m,n})$-th off-diagonal entry of $\boldsymbol{\Sigma}_{ \hat{g}}$ is written as  $\sum_{l=1}^{\epsilon \text{-rank}} s_{\textcolor{black}{m}} u_{\textcolor{black}{m},l} u_{\textcolor{black}{n},l} $. \textcolor{black}{Therefore,} the more decaying are the eigenvalues $s_{\textcolor{black}{m}}$ for $\textcolor{black}{m} > \epsilon \text{-rank}$, the more negligible is their contribution to the sum $\sum_{l=1}^{N} s_{\textcolor{black}{m}} u_{\textcolor{black}{m},l} u_{\textcolor{black}{n},l} $. Consequently, $\sum_{l=1}^{\epsilon \text{-rank}} s_{\textcolor{black}{m}} u_{\textcolor{black}{m},l} u_{\textcolor{black}{n},l} \approx \sum_{l=1}^N s_{\textcolor{black}{m}} u_{\textcolor{black}{m},l} u_{\textcolor{black}{n},l}$ and  $\boldsymbol{\Sigma}_{ \hat{g}} \approx \boldsymbol{\Sigma}_g$, which improves the accuracy of the first stage approximation.

\textcolor{black}{In Fig. \ref{eigenvalues}, we plot $\frac{\textcolor{black}{\mathrm{\#}}(k:~s_k > \epsilon)}{N}$ versus $\epsilon$}, for $\sigma =1$, and for multiple values of $N$ and $W$ to illustrate the decaying profile of the eigenvalues. For instance, for $N=200$ and $W = 0.2$, $P(s_k > 3\times 10^{-15})$ is less than $0.045$, which means that more than $95 \%$ of the eigenvalues of $\boldsymbol{\Sigma}_g$ are less than $3\times 10^{-15}$. In this case, it makes sense to take $\epsilon\text{-rank} = \lfloor N \times 0.045 \rfloor = 9$ to ensures \textcolor{black}{a satisfactory approximation}. In other words, the first stage approximation considers only the contribution of the dominant eigenvalues of the covariance matrix, which well captures the correlations between the ports.

Fig. \ref{eigenvalues} also demonstrates that the percentage of dominant eigenvalues decreases as $N$ increases and increases as $W$ increases. For example, for a fixed $W=2$, the eigenvalues exceeding $7.5 \times 10^{-15}$ decreased from $30 \%$ to $8 \%$ when $N$ increased from $50$ to $200$ . On the other hand, they increased from $6.5 \%$ to $8 \%$ when $W$ increased from $1$ to $2$ for a fixed $N=200$. We recall that the length of the antenna array is $W \lambda$, and the spacing between every two ports is $\frac{W \lambda}{N-1}$. Therefore, if we fix $N$, the larger is $W$, the more decorrelated and spaced are the FAS ports. In this case, the number of dominant eigenvalues keeps increasing as we increase $W$ to reach $N$ when the covariance matrix becomes diagonal with diagonal elements (i.e. eigenvalues) equal to $\sigma^2$. Furthermore, if we fix $W$, the larger $N$, the closer and more correlated the FAS ports, and the smaller the number of dominant eigenvalues of the covariance matrix. 

To conclude, \textcolor{black}{Fig. \ref{Cov_model_comp}} and \textcolor{black}{F}ig. \ref{eigenvalues} show that \textcolor{black}{Jake's correlation model can have a considerable percentage of negligible eigenvalues when a large number of ports is considered in a small space. In this case,} taking $\epsilon\text{-rank}$ as the number of dominant eigenvalues makes $\epsilon\text{-rank} \ll N$ and ensures a considerable reduction in the number of multi-fold integrals (\textcolor{black}{P}roposition \ref{gain}), \textcolor{black}{while guaranteeing a satisfactory approximation.}  
\begin{figure}[t]
    \centering
    \includegraphics[width=0.8 \textwidth]{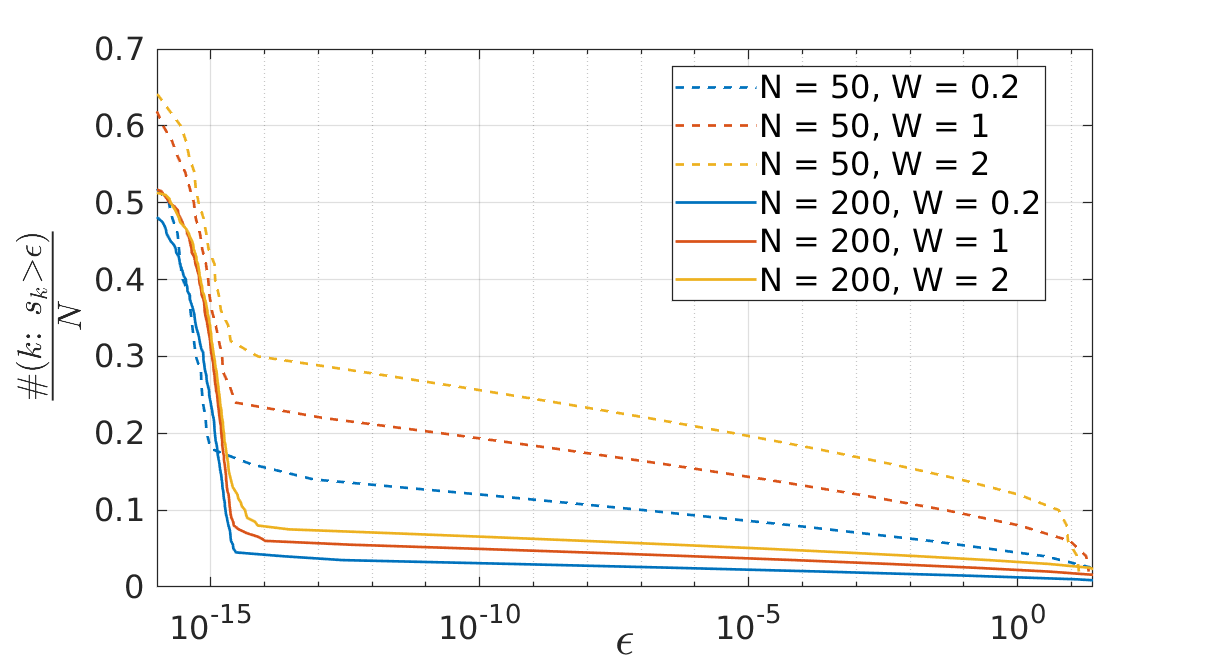}
    \caption{
    \textcolor{black}{Eigenvalues distribution of $\boldsymbol{\Sigma}_g$ for $\sigma=1$ and different values of $N$ and $W$.} }
    \label{eigenvalues}
\end{figure}

\paragraph{First Stage Approximation} \textcolor{black}{F}ig. \ref{fig:epsilonrank} shows \textcolor{black}{Monte Carlo} simulations of the FAS channel versus the first stage approximation model for $W=1$, $N=100$ and $\textcolor{black}{\sigma} = 10$. We compare their empirical CDFs to investigate the influence of $\epsilon$ on the accuracy of the first stage approximation. We can see that the approximation improves as $\epsilon$ decreases (\textit{i.e.} as $\epsilon$-rank increases), which confirms the result of \textcolor{black}{T}heorem \ref{convergence}. Furthermore, we can see that high accuracy is obtained by only taking $\epsilon$-rank $= 5 \ll N$. Therefore, according to \textcolor{black}{P}roposition \ref{gain}, the number of the multi-fold integrals in the approximated probability of outage is reduced by $0.9$ (from 100 integrals to 10 integrals), which is a considerable computational gain.
\begin{figure}[t]
    \centering
    \includegraphics[width = 0.8\linewidth]{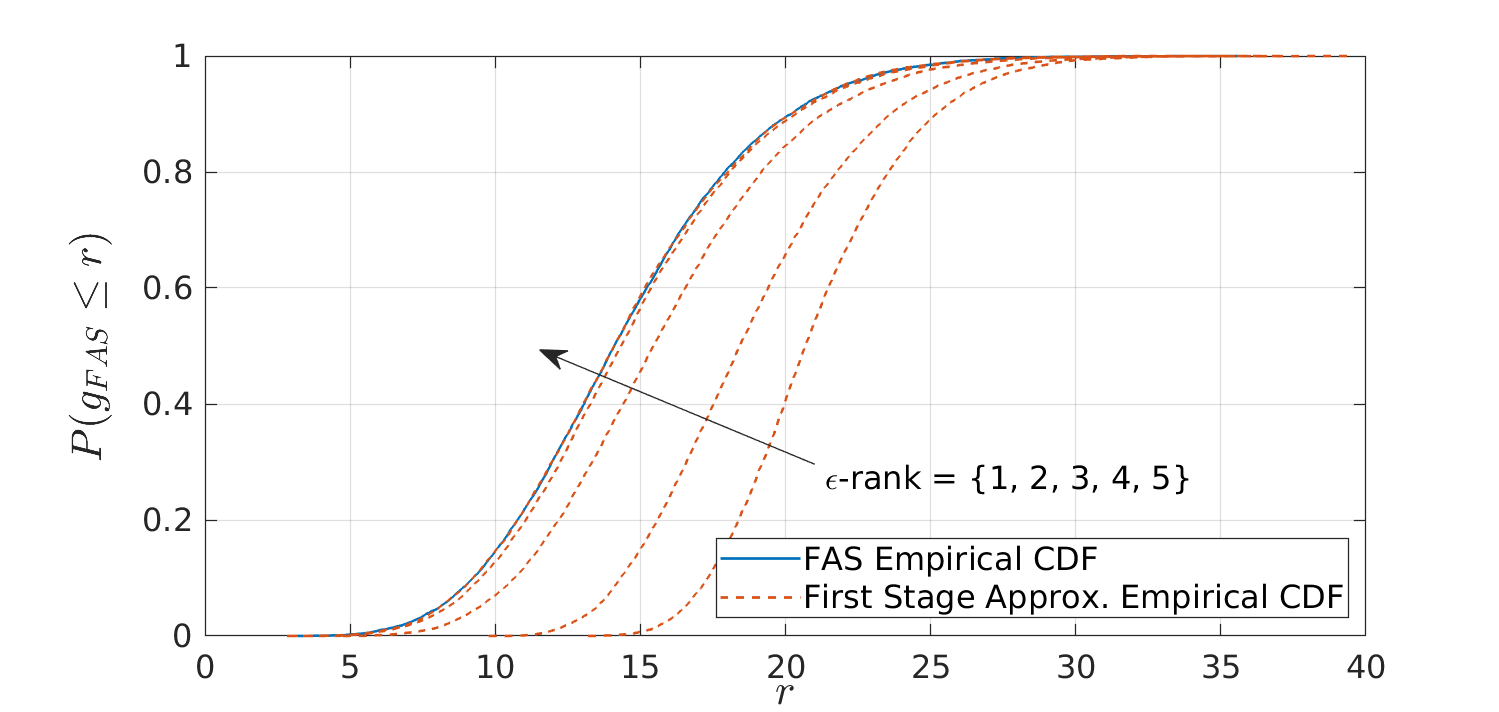}
    \caption{Empirical CDF of FAS versus first stage approximation. }
    \label{fig:epsilonrank}
\end{figure}

The parameter in question, as far as the first approximation is concerned, is $\epsilon$-rank. It can be determined numerically by counting the number of eigenvalues of $\textcolor{black}{\boldsymbol{\Sigma}}_g$ that exceed a certain threshold $\epsilon$ for given $N$ and $W$. However, for more insight, we investigate how $\epsilon$-rank varies as a function of the problem parameters $N$ and $W$. First, we determine an asymptotic expression of $\epsilon$-rank for a large $N$, given by \textcolor{black}{T}heorem \ref{theo:epsilonrank}. Then, we choose a threshold $\epsilon = \frac{\sigma^2}{2N}$. In other words, $\epsilon$-rank will be the number of eigenvalues of $\boldsymbol{\Sigma}_g$ exceeding $\frac{\sigma^2}{2N}$, and the eigenvalues less than $\frac{\sigma^2}{2N}$ will be considered negligible in the first stage approximation. \textcolor{black}{Thus,} we have on one hand that $\frac{N-1}{\pi W} > \frac{2}{\pi} \geq 0.63$ since $\frac{W}{N-1} < \frac{1}{2}$. On the other hand, $\frac{1}{2N} < 0.5$. Therefore, $\epsilon = \frac{\sigma^2}{2N} < \sigma^2 \frac{N-1}{\pi W}$, and according to \textcolor{black}{T}heorem \ref{theo:epsilonrank}, the expression of $\epsilon$-rank becomes independent of $\epsilon$, and can be approximated by $\ceil{2 W \frac{N}{N-1}}$ for a large $N$, as in (\ref{epsilon-rank-approx0}). 

For more precision around the smaller values of $N$, we propose to approximate $\epsilon\text{-rank}$ as $ \ceil{a W \frac{N}{N-1}}$ where $a$ is a parameter that we determine numerically. We consider the $N\times N$ matrix $\textbf{T}_N$ defined such that $ (\textbf{T}_N)_{\textcolor{black}{m,n}} = J_0 \left( \frac{2 \pi W (\textcolor{black}{m-n}) }{N-1} \right) $ $\text{for }1\leq \textcolor{black}{m,n}\leq N$. On one hand, we have $\textbf{T}_N = \textcolor{black}{\boldsymbol{\Sigma}}_g$ for $\sigma=1$. On the \textcolor{black}{other hand}, our choice of threshold $\epsilon = \frac{\sigma^2}{2N}$ ensures that $\epsilon$-rank does not depend on $\sigma^2$ (\textcolor{black}{T}heorem \ref{theo:epsilonrank}). Therefore, finding $\epsilon$-rank of $\textcolor{black}{\boldsymbol{\Sigma}}_g$ for an arbitrary $\sigma^2$ is equivalent to finding $\epsilon$-rank of $\textbf{T}_N$. To determine $a$ numerically, we simply count the number of eigenvalues exceeding $\epsilon = \frac{1}{2N}$, for $N \in \{ 10,\ldots, 300 \}$ and $W\in [0.1,5]$. Then, we determine $a$ that minimizes the mean squared error (MSE) between the actual $\epsilon\text{-rank}$ and the approximated $\epsilon\text{-rank}$ given by the formula $\epsilon\text{-rank}= \ceil{a W \frac{N}{N-1}}$. The minimization of the MSE gives $a= 3.1935 $. Therefore, $\epsilon\text{-rank} \approx \ceil{3.1935 W \frac{N}{N-1}}$, which has been validated by simulation results. For instance, in \textcolor{black}{F}ig. \ref{fig:epsilonrank}, we can see that taking $\epsilon\text{-rank}= 4$ guarantees a satisfactory approximation of the CDF, which is also given by $\ceil{3.1935 W \frac{N}{N-1}}$ for $N=100$ and $W=1$. We also can see from the approximated expression of $\epsilon\text{-rank}$ that it varies very fast with $W$ and slowly with $N$. This can also be observed in \textcolor{black}{F}ig. \ref{eigenvalues}. Further, we can see that for our target ranges of $N \in  \{ 10,\ldots, 300 \}$ and $W\in [0.1,5]$, $\ceil{3.1935 W \frac{N}{N-1}} \ll N$. Therefore, the values taken by $\epsilon$-rank are very small compared to $N$. This again guarantees a considerable reduction in the number of multi-fold integrals of the approximation. 

\begin{figure}[t]
    \centering
    \includegraphics[width = 0.8\linewidth]{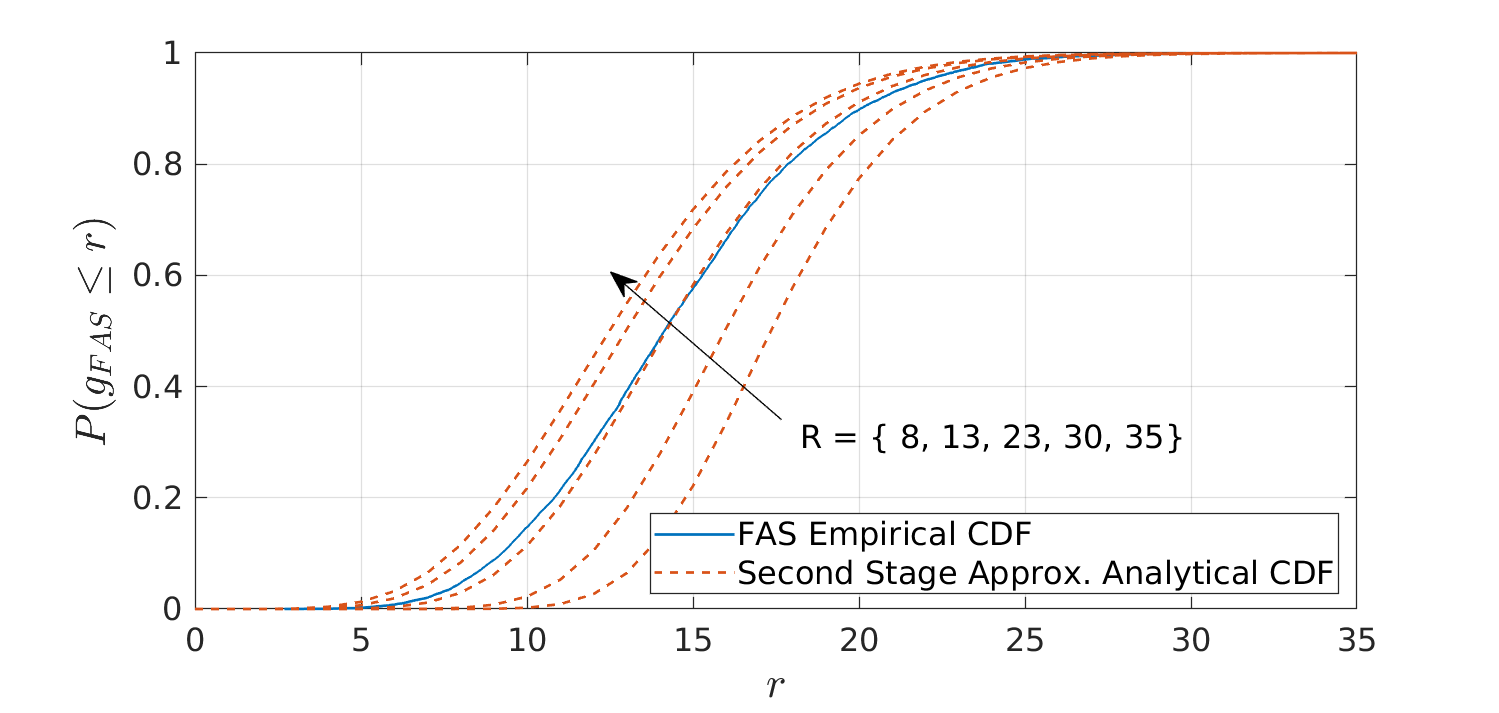}
    \caption{Empirical CDF of FAS versus second stage approximation. }
    \label{fig:R}
\end{figure}

\paragraph{Second Stage Approximation} 

To further improve mathematical tractability, we design the second stage approximation that provides a single-integral expression of the CDF as in \textcolor{black}{T}heorem \ref{theo:joint_distribution_g_approx}, and we assess its accuracy. Fig. \ref{fig:R} shows \textcolor{black}{Monte Carlo simulation} of the FAS channel versus the analytical CDF expression given by the second stage approximation for $W=1$, $N=100$ and $\textcolor{black}{\sigma} = 10$. We can see that the approximation starts improving as $R$ increases from $8$ to $13$. It reaches a satisfactory approximation for $R=23$, and as $R$ keeps increasing from $30$ to $35$, the approximation starts degrading. Therefore, we investigate the optimal value of $R$ that gives the best approximation as a function of the problem parameters. An approximated optimal value is given by  $R^* = \min \left \{ \left \lfloor\frac{1.52(N-1)}{2\pi W}\right \rfloor, N \right \}$. Although the proof of optimality is provided for the relaxed problem (\ref{P3}), simulation results show that $R^*$ provides a satisfactory approximation for the general optimization problem (\ref{P1}). For instance, taking the case of \textcolor{black}{F}ig. \ref{fig:R}, we have $R^* = \min \left \{ \left \lfloor\frac{1.52\times 99}{2\pi}\right \rfloor, 100 \right \} = 23$.

The second stage approximation provides a satisfactory single-integral approximation of the CDF of the FAS channel. Therefore, we compare our second stage approximation CDF to the analytical CDF in \cite[eq.
(16)]{FAS1}, and the empirical CDF of the FAS channel. In \textcolor{black}{F}ig. \ref{fig:W1N40and200}, we take $\sigma = 10$, we consider $N=40$ and $N=200$ and we fix $W=1$. By examining the compared CDFs in \textcolor{black}{F}ig. \ref{fig:W1N40and200}, we can see two main observations. First, our model provides a more accurate approximation of the FAS empirical CDF. In fact, we can see that the model from \cite{FAS1} is optimistic because it lower-bounds the empirical CDF of the FAS channel, and therefore, it lower-bounds the probability of outage. Second, we can see that as $N$ varies from $40$ to $200$, the FAS empirical CDF, as well as our approximation, remain almost unchanged. This may indicate that the probability of outage does not necessarily decrease as $N$ increases. However, when $N$ increases from $40$ to $200$, the analytical CDF in \cite{FAS1} significantly shifts to the right, indicating an optimistic decrease in the probability of outage. In \textcolor{black}{F}ig. \ref{fig:N200W1and4}, we take $\sigma = 10$, we consider $W=1$ and $W=4$ and we fix $N=200$. By examining the compared CDFs, we can see again that our second stage approximation of the CDF provides a more accurate result of the empirical CDF of the FAS channel than the analytical CDF in \cite{FAS1}. Furthermore, we can see that as $W$ varies from $1$ to $4$, the CDF in \cite{FAS1} remains almost unchanged while the FAS empirical CDF, as well as our approximation, show a considerable shift to the right. This indicates that for $N=200$, increasing the length of the FAS line space can have a major impact on the probability of outage, which is not captured by the model in \cite{FAS1}. 

To conclude, these experiments show that, for a fixed $N$, increasing the spatial separation between the ports by increasing $W$, and thus decreasing their inter-correlation, can significantly reduce the probability of outage. On the other hand, if we fix $W$, decreasing the spatial separation between the FAS ports by increasing $N$, and thus increasing their inter-correlation, does not seem to have a major impact on the probability of outage. 
\begin{figure}[t]
\begin{subfigure}{.5\textwidth}
  \centering
  \includegraphics[width=\linewidth]{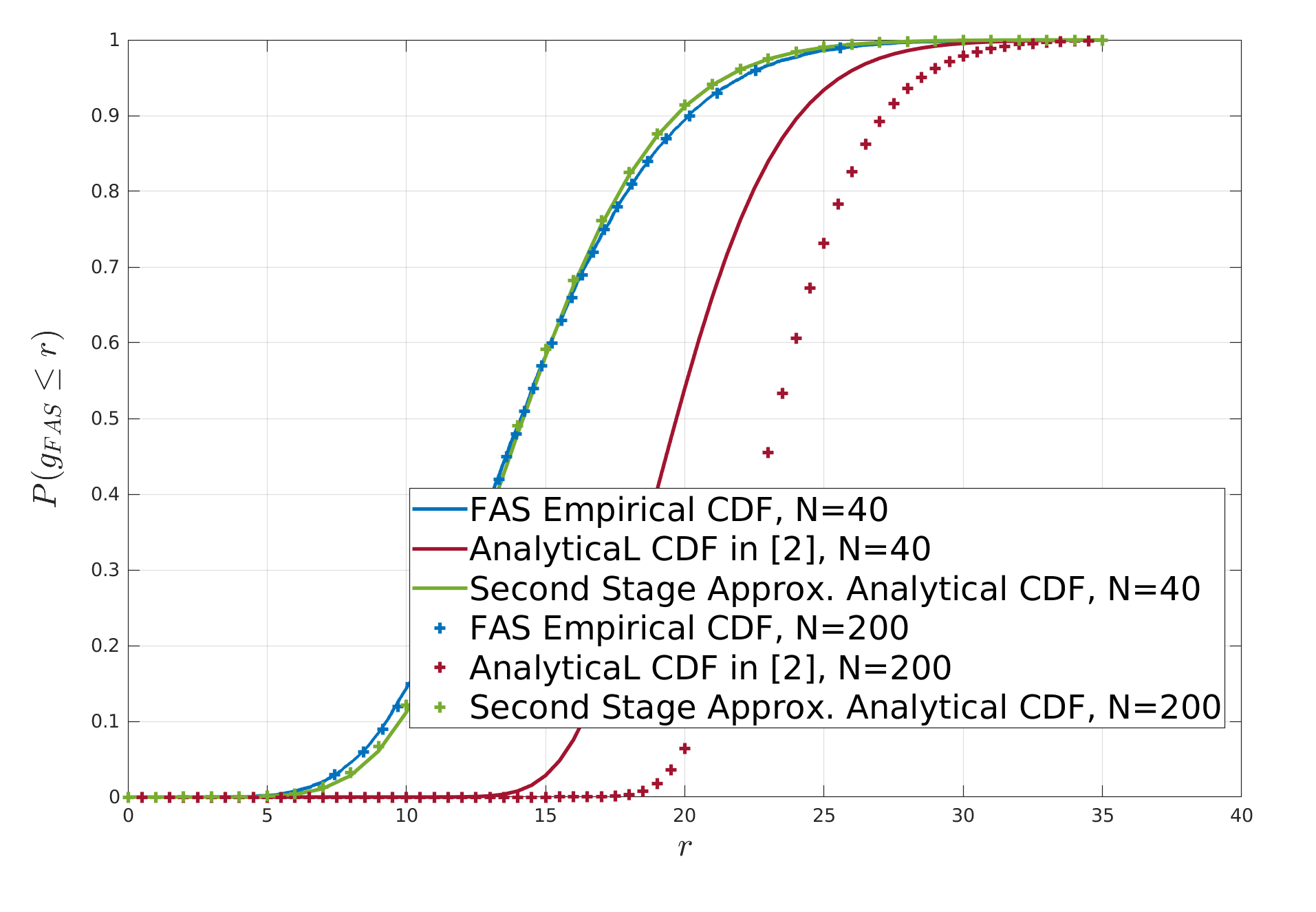}  
  \caption{CDF for fixed $W=1$ and $\sigma = 10$.}
  \label{fig:W1N40and200}
\end{subfigure}
\begin{subfigure}{.5\textwidth}
  \centering
  \includegraphics[width=\linewidth]{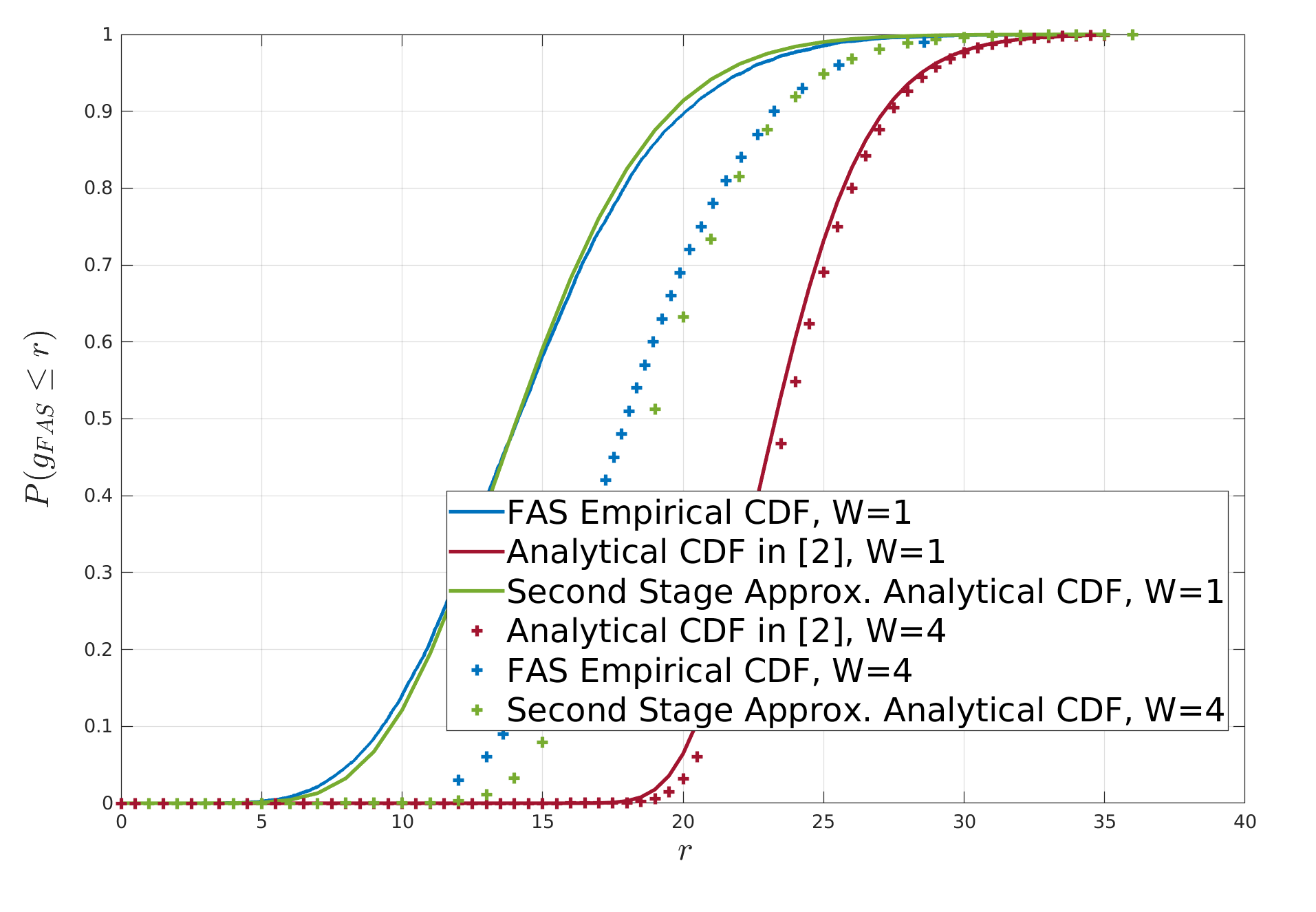}  
  \caption{CDF for fixed $N=200$ and $\sigma = 10$.}
  \label{fig:N200W1and4}
\end{subfigure}
\caption{\textcolor{black}{Empirical CDF of FAS versus second stage approximation CDF and CDF in \cite{FAS1}}.}
\label{fig:approximation_vs_1}
\end{figure}
These two observations do not match the results given by \cite{FAS1} for the tested values of $N$ and $W$ in \textcolor{black}{F}ig. \ref{fig:approximation_vs_1}. Nevertheless, it is important to highlight that we are considering a relatively high density of ports. Indeed, this density does not only depend on the $W$ to $N$ ratio, but it also depends on the wavelength $\lambda$, since the inter-ports spacing is $\frac{W \lambda}{N-1}$. According to \cite{FAS1}, sub-6GHz is particularly suitable for FAS because the Rayleigh fading model becomes inaccurate for the millimeter-wave bands. If we consider the frequency $5$GHz, we will have a port approximately every 1.2 millimeters for $W=1$, $N=40$ and $W=4$, $N=200$. In the case of $W=1$, $N=200$, the inter-antenna spacing can go down to 0.3 millimeters. In these cases, \textcolor{black}{F}ig. \ref{fig:approximation_vs_1}, show that the CDF of the FAS \cite{FAS1} fails to match the empirical cdf of the FAS. However, for a few spaced ports, the CDF of the FAS in \cite{FAS1} approaches more the empirical CDF, and it has been shown to be exact for $N=2$.

\begin{figure}[t]
\begin{subfigure}{.49\textwidth}
  \centering
  \includegraphics[width=\linewidth]{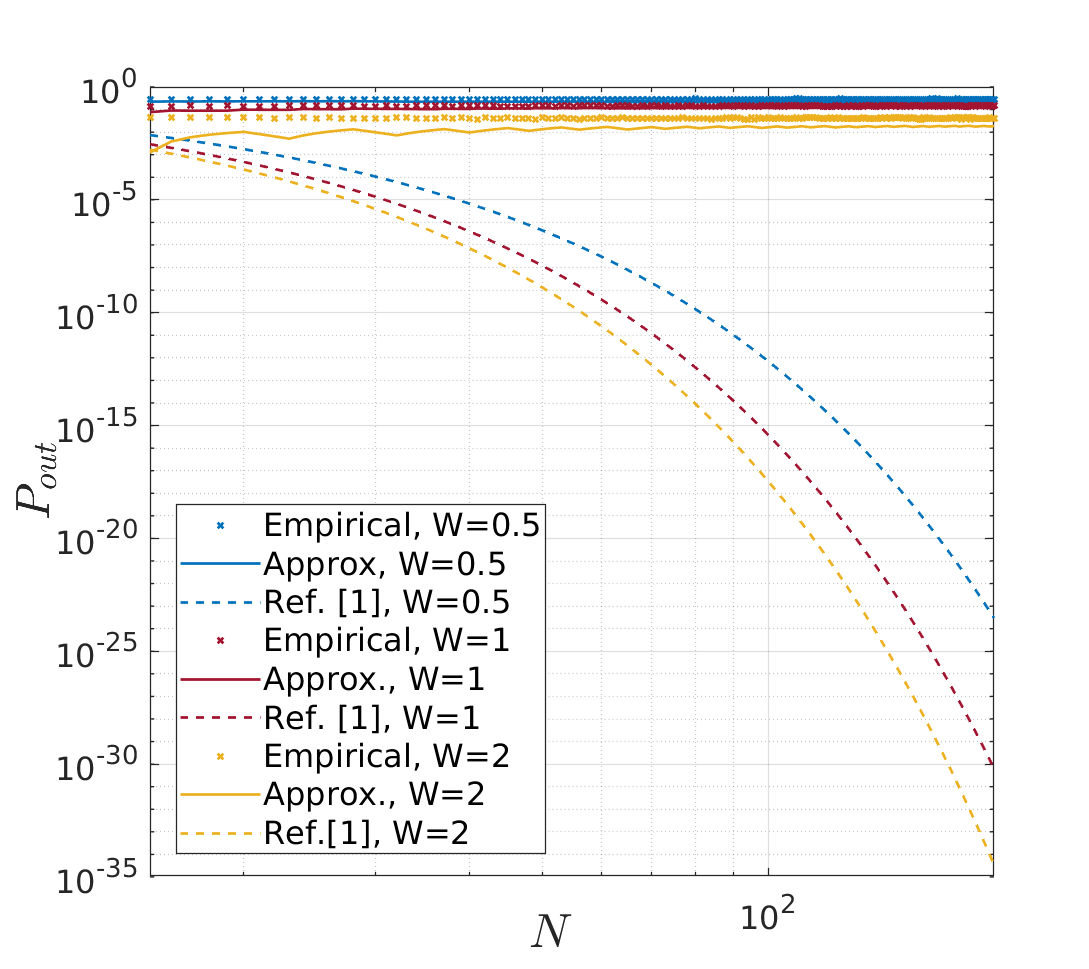}  
  \caption{Probability of Outage for $\gamma_{th}/\Gamma = 0$~dB.}
  \label{fig:poutvsN1}
\end{subfigure}
\begin{subfigure}{.49\textwidth}
  \centering
  \includegraphics[width=\linewidth]{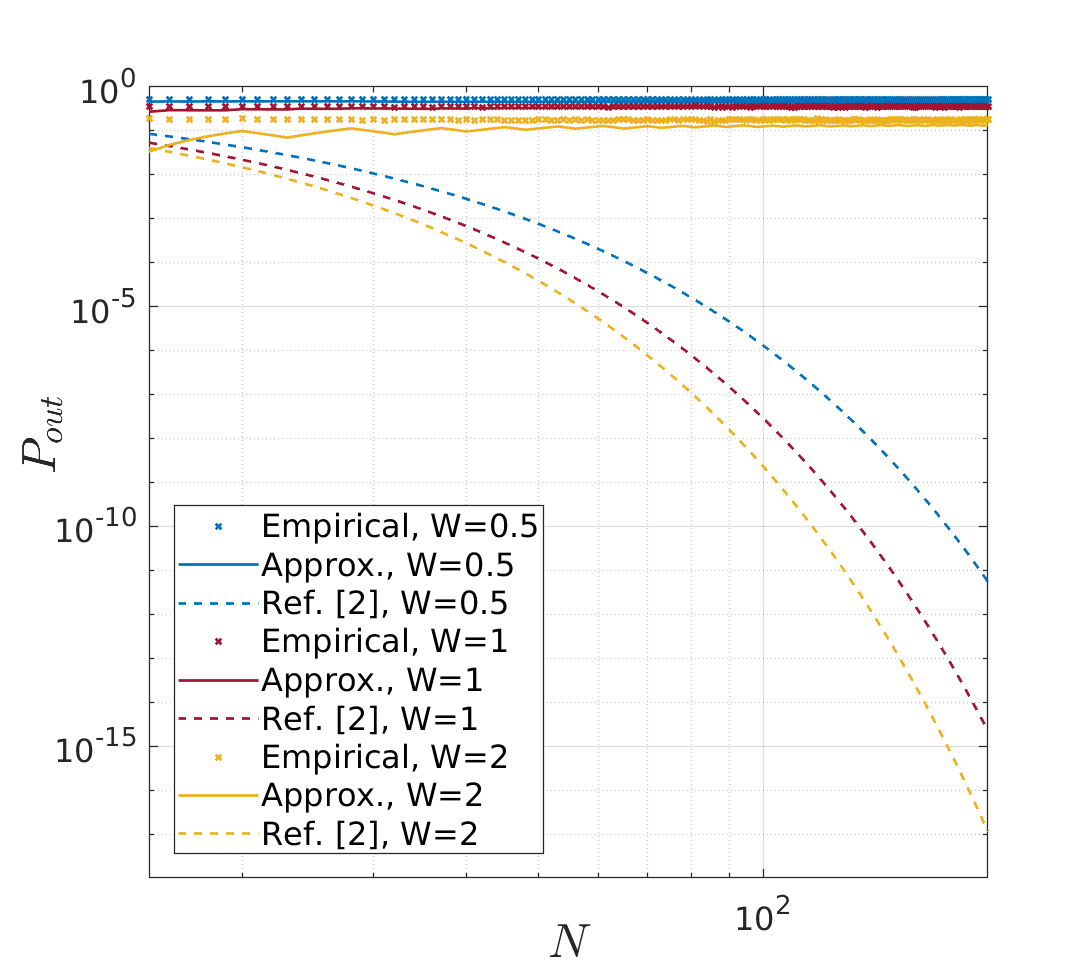}  
  \caption{Probability of Outage for $\gamma_{th}/\Gamma = 2$~dB.}
  \label{fig:poutvsN2}
\end{subfigure}\\[1ex]
\centering
\begin{subfigure}{0.49\textwidth}
  \centering
  \includegraphics[width=\linewidth]{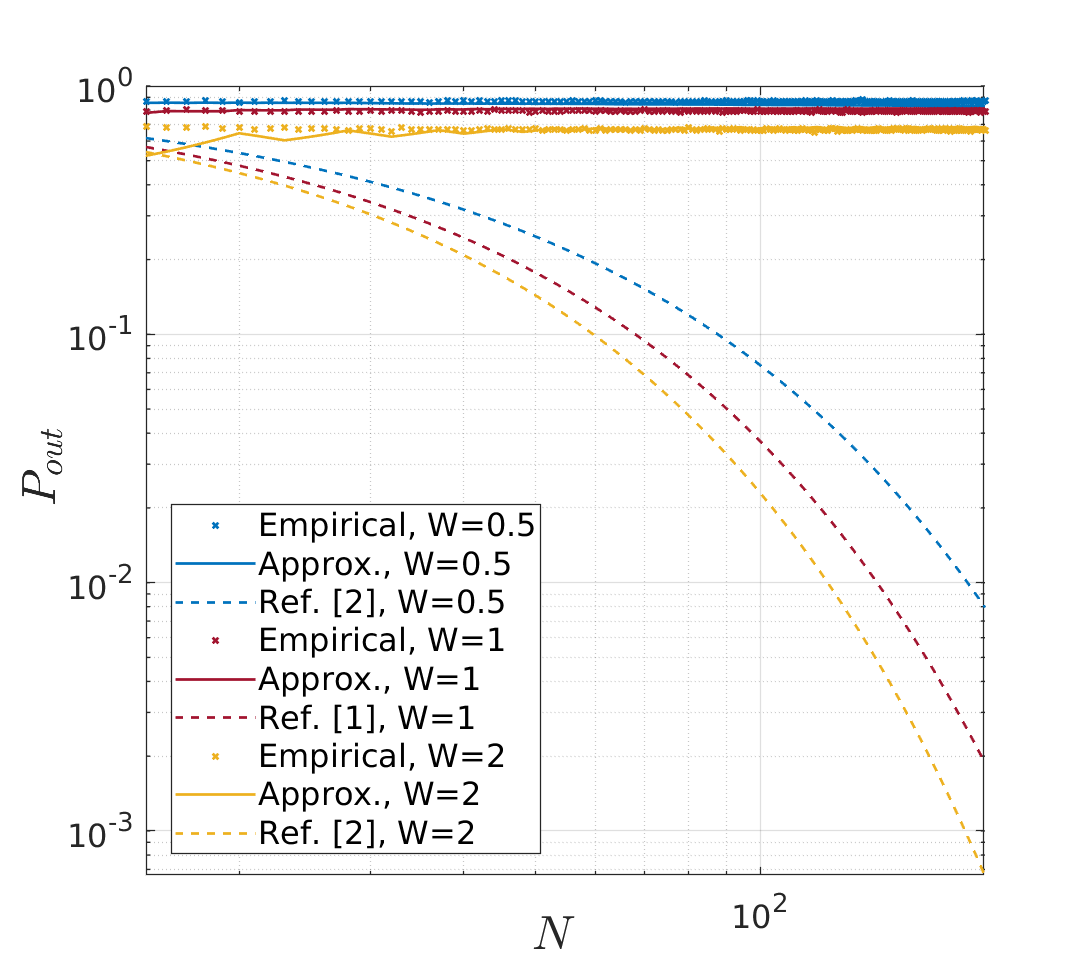}  
  \caption{Probability of Outage for $\gamma_{th}/\Gamma = 5$~dB.}
  \label{fig:poutvsN3}
\end{subfigure}
\caption{ \textcolor{black}{Empirical $P_{out}$ of FAS versus second stage approximation $P_{out}$ and $P_{out}$ in \cite{FAS1}.}}
\label{fig:poutvsN}
\end{figure}

\paragraph{FAS performance Analysis} \textcolor{black}{T}o assess the FAS performance for a larger range of $N$, we plot the outage probability $P_{out}$ versus the number of ports $N$, for different values of target SNR $\gamma_{th}/\Gamma$, and different values of $W$. More precisely, in \textcolor{black}{F}ig. \ref{fig:poutvsN}, we compare the probability of outage given by the second stage approximation (Approx.) to the one given in \cite[eq. (16)]{FAS1} (Ref. \cite{FAS1}), and to the empirical probability of outage (Empirical) \textcolor{black}{obtained by $100000$ draws of Monte Carlo simulations}. 

In \textcolor{black}{Fig. \ref{fig:poutvsN1}, Fig. \ref{fig:poutvsN2}, and Fig. \ref{fig:poutvsN3} we fix $\gamma_{th}/\Gamma=0$dB, $\gamma_{th}/\Gamma=2$dB, and $\gamma_{th}/\Gamma=5$dB respectively. Further, for each of the three scenarios, we consider three values of $W$, namely, 0.5, 1, and 2. As expected, the probability of outage rises as the SNR target increases and drops as the space increases with $W$. However, the main observation is that the probability of outage of our model closely approximates the empirical probability of outage, which is considerably high and almost constant compared to the one in \cite{FAS1}. More precisely, if we consider the SNR target $\gamma_{th}/ \Gamma=0$dB and $W=1$, both the empirical and approximated probabilities of outage are around $10^{-1}$ independently of the number of ports $N$. On the other hand, the probability of outage in \cite{FAS1} decreases without a floor to reach $1.52~10^{-23}$ for $N=150$. This shows that parameterizing the channel as in \cite{FAS1} does not accurately capture the correlation between the FAS ports, providing an optimistic performance analysis. In fact, while it ensures that the correlation between the reference port and the other ports follows Jake's model, the other correlations do not follow the covariance function given by Jake's model. Therefore, simplifying the correlation structure of these highly dependent channels can considerably influence the achievable performance analysis.}

\begin{figure}[t]
    \centering
    \includegraphics[width = 0.8\linewidth]{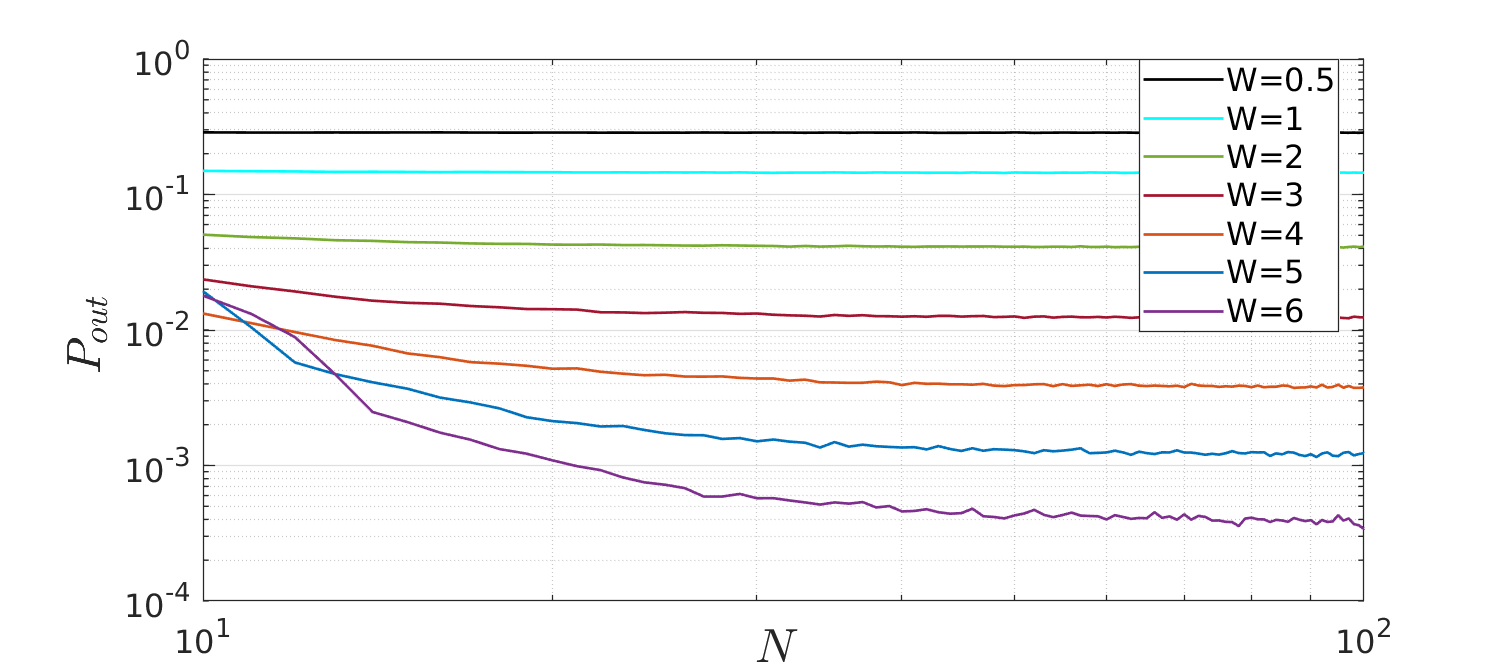}
    \caption{\textcolor{black}{Empirical $P_{out}$ of the FAS channel for $\gamma_{th}/\Gamma=0$dB.} }
    \label{fig:diversity_gain}
\end{figure} 

\textcolor{black}{ In Fig. \ref{fig:poutvsN}, we consider the same values of $W$ as in \cite{FAS1} to carry out the comparison between the different results in the same settings of \cite{FAS1}. However, it is important to highlight that the considered values of $W$ can be relatively small. This implies that the ports are in a narrow space (this also depends on the wavelength $\lambda$), resulting in a low level of diversity. On the other hand, if $W$ is larger, there might still be non-negligible diversity to exploit by having more ports. To investigate this, we plot the empirical probability of outage of the FAS channel, in Fig. \ref{fig:diversity_gain}, for $\gamma_{th}/\Gamma=0$dB, and different values of $W$ that go up to $6$. We can see that for small values of $W$ such as $0.5$, $1$, and $2$, the probability of outage remains almost constant as $N$ increases. However, for the larger values of $W$, the probability of outage decreases first as $N$ increases, indicating a diversity gain, and then it saturates. Furthermore, the probability of outage saturates at larger values of $N$ as $W$ increases. In other words, as $N$ rises for a fixed $W$, we are reducing the space between the ports, increasing their inter-correlation, and therefore, lowering diversity gain until saturation. Therefore, unlike \cite{FAS1}, we show that the probability of outage does not drop without a floor as $N$ increases but rather saturates. However, there still is a diversity gain in a small space, as it is shown for the larger values of $W$ in Fig. \ref{fig:diversity_gain}. However, it is constrained by space and needs careful analysis to investigate the specific ranges of $N$ and $W$ where we can have a significant gain with respect to traditional multi-antenna systems. 
}

\section{Conclusion}
Fluid antenna systems allow for adopting multiple antennas in a mobile device by exploiting diversity hidden in a small space, and previous related works have revealed considerable performance gain. Nevertheless, assessing the true benefits of such arising technology remains a function of the careful modeling of the FAS channel. \textcolor{black}{It is, however, challenging to model the channel in an analytically tractable manner due to the high correlation between the FAS ports. Therefore, this work proposes a two-stage approximation of the FAS channel. Unlike previous works in the literature, where the correlation model of the channel is rather simplified, the proposed approximations incorporate more parameters into the channel model that are carefully chosen to follow Jake's correlation model closely. To this end, our approximation offers a more accurate look at the achievable performance of FAS by capturing the correlation effect on diversity gain. The paper presents the first stage approximation, which considerably reduces the number of multi-fold integrals in outage probability. Then, we provide a second-stage approximation representing the probability of outage as a power of a single integral.} Finally, our numerical results assess the accuracy of the proposed approximations and compare our work with the previous related work. In our results, \textcolor{black}{ we argue that the probability of outage of the FAS channel does not decrease without a floor as $N$ increases, but rather saturates, which constrains the FAS performance by space limitations. This is mainly due to considering a less idealized correlation model. However, various other practical considerations remain open topics including the impact of delay and frequency deviations on CSI estimation at the ports. Therefore, our work opens the door for further research to address more implementation limitations of FAS and to investigate designs where diversity gain can be guaranteed.}

\appendices

\section{} \label{appendixA}

\textcolor{black}{From the diagonalization of the covariance matrix $\boldsymbol{\Sigma}_g$, we can write 
\begin{equation}
    (\boldsymbol{\Sigma}_g)_{m,n} = \sum_{l=1}^{N} s_l u_{m,l} u_{n,l}, ~ \forall m,n \in \{1,\ldots, N\}
    \label{eq:sigma_g}
\end{equation}}
\noindent \textcolor{black}{ For $\sigma_h = \sigma$, $M=N$ and $\alpha_{m,n} = \frac{\sqrt{s_n}}{\sigma} u_{m,n},~ 1\leq m,n \leq N $, (\ref{cov_h}) and (\ref{eq:sigma_g}) imply that
\begin{equation}
    (\textcolor{black}{\boldsymbol{\Sigma}}_g)_{m,n} =  (\textcolor{black}{\boldsymbol{\Sigma}}_h)_{m,n},~~ \text{for } m,n \in \{1,\ldots, N\}. 
    \label{same_dist}
\end{equation}
 Since $\textbf{g}$ and $\textbf{h}$ have zero mean vectors and the same covariance matrix according to (\ref{same_dist}), then they have the same joint distribution. }

\noindent  \textcolor{black}{Furthermore, we have
\begin{equation}
    \sum_{n=1}^{M} \alpha_{m, n}^{2} = \sum_{n=1}^{N} \frac{s_n}{\sigma^2} u_{m,n}^{2} = 1.
\end{equation}
Therefore, The representation in (\ref{eq:theo12}) follows immediately by substituting the parameters. }


\section{} \label{appendixB}

\textcolor{black}{ 
Let $\textbf{D}=\mathrm{diag}(\sum_{k =\epsilon\text{-rank}+1}^N s_k u_{1,k}^2,\ldots,\sum_{k =\epsilon\text{-rank}+1}^N s_k u_{N,k}^2)$, $\textbf{U}=(\textbf{u}_1, \ldots, \textbf{u}_N)$, $\textbf{S}=\mathrm{diag}(s_1,\ldots,s_N)$, $\textbf{S}'=\mathrm{diag}(s_1,\ldots,s_{\epsilon\text{-rank}},0,\ldots, 0)$.}

\noindent \textcolor{black}{ 
We can write $\boldsymbol{\Sigma}_g = \textbf{U} \textbf{S} \textbf{U}^{\mathrm{T}}$ and $\boldsymbol{\Sigma}_{\hat{g}} = \textbf{U} \textbf{S}' \textbf{U}^{\mathrm{T}} + \textbf{D}= \boldsymbol{\Sigma}'+ \textbf{D}$. According to \cite{FrechetDistance}, we have 
\begin{align}
\begin{split}
     W_2(\mathcal{CN}(\textbf{0}_{N \times 1}, \boldsymbol{\Sigma}_g),\mathcal{CN}(\textbf{0}_{N \times 1}, \boldsymbol{\Sigma}_{\hat{g}})) 
     &= \mathrm{Tr}( \boldsymbol{\Sigma}_g + \boldsymbol{\Sigma}_{\hat{g}} - 2 \left(\boldsymbol{\Sigma}_g \boldsymbol{\Sigma}_{\hat{g}} \right)^{\frac{1}{2}})\\
     & \leq \mathrm{Tr}( \boldsymbol{\Sigma}_g + \boldsymbol{\Sigma}' + \epsilon \textbf{I} - 2 \left(\boldsymbol{\Sigma}_g (\boldsymbol{\Sigma}'+ \textbf{D}) \right)^{\frac{1}{2}})\\
     & \leq N \epsilon + \mathrm{Tr}( \boldsymbol{\Sigma}_g + \boldsymbol{\Sigma}'  - 2 \left(\boldsymbol{\Sigma}_g \boldsymbol{\Sigma}' \right)^{\frac{1}{2}})\\
     & \leq N \epsilon + || \left(\boldsymbol{\Sigma}_g \right)^{\frac{1}{2}}  -  \left( \boldsymbol{\Sigma}' \right)^{\frac{1}{2}})||_\mathrm{F} \text{, because $\boldsymbol{\Sigma}'$ and $\boldsymbol{\Sigma}_g$ commute}\\
     &= N \epsilon + || \left( \textbf{S} \right)^{\frac{1}{2}}  -  \left( \textbf{S}' \right)^{\frac{1}{2}})||_\mathrm{F}\\
     & \leq N \epsilon + (N - \epsilon\text{-rank}) \epsilon^2.
\end{split}
\end{align}
}


\section{} \label{appendixC}
\textcolor{black}{For $\epsilon > 0$ and $\epsilon\text{-rank} < N$, let $\textbf{a}$ and $\textbf{b}$ be the random vectors defined as 
\begin{align}
    \textbf{a}=(a_1, a_2,\dots, a_{\epsilon\text{-rank}})^{\textcolor{black}{\mathrm{T}}}, ~~~~ \textbf{b}=(b_1, b_2,\dots, b_{\epsilon\text{-rank}})^{\textcolor{black}{\mathrm{T}}}.
\end{align}}
\noindent \textcolor{black}{We can see that $\hat{g}_1|(\textbf{a},\textbf{b}), \hat{g}_2|(\textbf{a},\textbf{b}),\ldots, \hat{g}_N|(\textbf{a},\textbf{b})$ are independent and we have 
\begin{equation}
    \hat{g}_k |(\textbf{a},\textbf{b}) \sim \mathcal{CN} \left( \sum_{l=1}^{\epsilon\text{-rank}} \sqrt{s_l} u_{k,l}(a_l +\textcolor{black}{\mathrm{j}} b_l), (\sigma^2-\sum_{l=1}^{\epsilon\text{-rank}} s_l u_{k,l}^2) \right),~~ \forall k \in \{1, 2,\ldots, N \},
\end{equation}}
\noindent \textcolor{black}{which implies that $|g_k|~|(\textbf{a},\textbf{b})$ is a Rician distribution \cite{stuber} and its CDF can be written as 
\begin{equation}
    F_{|\hat{g}_k|~|(\textbf{a},\textbf{b})}(r_k) = 1-Q_1 \left(\frac{\sqrt{2\left(\sum \limits_{l=1}^{\epsilon\text{-rank}} \sqrt{s_l} u_{k,l} a_l\right)^2+ 2\left(\sum \limits_{l=1}^{\epsilon\text{-rank}} \sqrt{s_l} u_{k,l} b_l\right)^2}}{\sqrt{\sigma^2 -\sum \limits_{l=1}^{\epsilon\text{-rank}} s_l u_{k, l}^{2}} }, \frac{\sqrt{2} r_k}{ \sqrt{\sigma^2 -\sum \limits_{l=1}^{\epsilon\text{-rank}} s_l u_{k, l}^{2}} }\right).
\end{equation}}
\noindent \textcolor{black}{By independence, we can write 
\begin{equation}
     F_{(|\hat{g}_1|, |\hat{g}_2|,\ldots, |\hat{g}_N|)|(\textbf{a},\textbf{b})}(r_1, r_2,\ldots, r_N) =  \prod_{k=1}^N F_{|\hat{g}_k|~|(\textbf{a},\textbf{b})}(r_k).
\end{equation}}
\noindent \textcolor{black}{Therefore, 
\begin{equation}
    F_{(|\hat{g}_1|, |\hat{g}_2|,\ldots, |\hat{g}_N|)}(r_1, r_2,\ldots, r_N) =  \idotsint\limits_{-\infty}^{\infty} \prod_{k=1}^N F_{|\hat{g}_k|~|(\textbf{a},\textbf{b})}(r_k) f_{\textbf{a},\textbf{b}} ~~ da_1\ldots da_{\epsilon\text{-rank}} ~db_1\ldots db_{\epsilon\text{-rank}},
    \label{proofappendixd}
\end{equation}}
\noindent \textcolor{black}{where $f_{a,b} (a_1,\ldots, a_{\epsilon\text{-rank}}, b_1,\ldots, b_{\epsilon\text{-rank}}) = \prod_{l=1}^{\epsilon\text{-rank}} \frac{1}{\pi} \exp{(-(a^2_l+b^2_l))}$.}
\noindent \textcolor{black}{Finally, considering that $F_{\max \{|\hat{g}_1|,\ldots,|\hat{g}_N|\}}(r) =  F_{|\hat{g}_1|,\ldots,|\hat{g}_N|} (r,\ldots,r)$, we get the expression in (\ref{cdfgfas}) by plugging $f_{a,b}$ and $F_{|\hat{g}_k|~|(a,b)}(r)$ in (\ref{proofappendixd}).}



\section{} \label{AppendixD}

We have $\textbf{T}_N$ of size $N \times N$ such that, for $ k,\ell \in \{1,\ldots,N\}$ and $0 < c < \frac{1}{2}$, 
\begin{align}
    \begin{split}
        (\textbf{T}_N)_{(k,\ell)} &= \sigma^2 J_0 \left( 2 \pi (k-\ell) c \right)\\
        &= \frac{\sigma^2}{2 \pi} \int_{-\pi}^{\pi} e^{i2 \pi(k-\ell) c \sin(\tau)} d\tau. 
    \end{split}
\end{align}

\noindent We consider the following Fourier transform pair, 
\begin{align}
\hat{f}(\lambda) &=\sum_{k=-\infty}^{\infty} \sigma^2 J_0 \left( 2 \pi k c \right) e^{i k \lambda} ; \lambda \in[-\pi , \pi]. \\
\sigma^2 J_0 \left( 2 \pi k c \right) &=\frac{1}{2 \pi} \int_{-\pi}^{\pi} \hat{f}(\lambda) e^{-i k \lambda} d \lambda.
\end{align}
\noindent Therefore, 
\begin{align}
    \begin{split}
        \hat{f}(\lambda) &=\sum_{k=-\infty}^{\infty} \frac{\sigma^2}{2 \pi} \int_{-\pi}^{\pi} e^{i2 \pi k c \sin(\tau)} d\tau e^{i k \lambda}\\
        &=  \frac{\sigma^2}{2 \pi} \int_{-\pi}^{\pi} \sum_{k=-\infty}^{\infty} e^{i 2 \pi k (c \sin{(\tau) + \frac{\lambda}{2 \pi}})} d\tau\\
        &= \frac{\sigma^2}{2 \pi} \int_{-\pi}^{\pi} \sum_{n=-\infty}^{\infty} \delta{(c \sin{(\tau)} + \frac{\lambda}{2 \pi} - n)} d\tau \text{ (by the Poisson sum formula \cite{lapidoth}).}\\
    \end{split}
\end{align}

\noindent For $\delta{(c \sin{(\tau)} + \frac{\lambda}{2 \pi} - n)}$ to be non zero, we need $c \sin{(\tau)} + \frac{\lambda}{2 \pi} - n$ to be zero. Thus,
\begin{align}
    c \sin{(\tau)} + \frac{\lambda}{2 \pi} - n = 0 &\Rightarrow  \sin{(\tau)} = \frac{1}{c} (n -\frac{\lambda}{2 \pi} ) \\
    &\Rightarrow -1 \leq \frac{1}{c} (n - \frac{\lambda}{2 \pi} ) \leq 1 \\
    &\Rightarrow \frac{\lambda}{2 \pi} - c  \leq n \leq \frac{\lambda}{2 \pi} + c. \label{cond_n}
\end{align}
\noindent Since $0 < c < \frac{1}{2}$ and $- \pi \leq \lambda \leq \pi$, then $\frac{\lambda}{2 \pi} + c < 1$ and $\frac{\lambda}{2 \pi} - c > -1$. Therefore, the only integer n that can verify (\ref{cond_n}) is zero. Therefore, 
\begin{align}
    \hat{f}(\lambda) &= \frac{\sigma^2}{2 \pi} \int_{-\pi}^{\pi} \delta{(c \sin{(\tau)} + \frac{\lambda}{2 \pi})} d\tau~ \mathbf{1}_{ \{0 \in [\frac{\lambda}{2 \pi}- c, \frac{\lambda}{2 \pi} + c ] \}}.
\end{align}
\noindent By using the identity \cite{diracdelta}, $ \delta( \Phi(\tau)) = \sum_{j} \frac{1}{|\Phi'(\tau_j)|} \delta(\tau- \tau_j)$ where $\tau_j$ is such $\Phi(\tau_j) = 0$ and $\Phi'(\tau_j) \neq 0$, we can write
\begin{align}
    \hat{f}(\lambda) = \frac{2 \sigma^2}{\sqrt{ (2\pi c)^2 - \lambda^2}} ~ \mathbf{1}_{ \{ -2\pi c \leq \lambda \leq 2\pi c \}}.
\end{align}
By applying the result from \cite{toeplitz}, corollary 5, we can write
\begin{align}
\begin{split}
    D(x) &= \frac{1}{2 \pi} \int_{\hat{f}(\lambda) \leq x} d\lambda \\
    &= \frac{1}{2 \pi} \int_{-\pi}^{\pi} \mathbf{1}_{ \{ \hat{f}(\lambda) \leq x \} } d\lambda \\
    &= \frac{1}{2 \pi} \int_{-\pi}^{-2 \pi c} \mathbf{1}_{ \{ 0 \leq x \} } d\lambda + \frac{1}{2 \pi} \int_{-2 \pi c}^{2 \pi c} \mathbf{1}_{ \{  \frac{2 \sigma^2}{\sqrt{ (2\pi c)^2 - \lambda^2}} \leq x \} } d\lambda + \frac{1}{2 \pi} \int_{2 \pi c}^{2 \pi} \mathbf{1}_{ \{ 0 \leq x \} } d\lambda \\ 
    &= \left \{ 
    \begin{array}{cc}
        1 - 2c & \text{ if } 0< x < \frac{\sigma^2}{\pi c} \\
        1 - 2c + \sqrt{ (2 c)^2 - \frac{4 \sigma^4}{(\pi x)^2}} & \text{ if } x \geq \frac{\sigma^2}{\pi c}.
    \end{array}
    \right.
\end{split}
\end{align}



\section{} \label{appendixE}

We can see that $\textcolor{black}{\mathrm{E}}(\hat{g}_{k,r}) = 0$ for all $k \in \{1,\ldots, N\}$ and $r \in \{1,\ldots, R\}$. Furthermore, for $i,j \in \{1,\ldots, N\}$ and $m,n \in \{1, \ldots, R\}$, 
\begin{align}
    \begin{split}
        (\textcolor{black}{\boldsymbol{\Sigma}}_{\hat{G}})^{m,n}_{i,j} &= \textcolor{black}{\mathrm{Cov}}(\hat{g}_{i,m}, \hat{g}_{j,n})
        =  \left\{
        \begin{array}{ll}
            0 & \mbox{if } m \neq n \\
            \sigma^2  & \mbox{if } i=j \text{ and } m=n \\
            \sum\limits^{\epsilon\text{-rank}}_{l=1} s_l u_{i,l} u_{j,l} & \mbox{if } i \neq j \text{ and } m=n,
        \end{array} 
        \right.
    \end{split}
\end{align}
 \noindent where $(\boldsymbol{\Sigma}_{\hat{G}})^{m,n}_{i,j}$ is the entry $(i,j)$ in the block $(m,n)$ of $\boldsymbol{\Sigma}_{\hat{G}}$. Therefore, the result in Proposition \ref{prop:covG_hat_tilde} follows directly from the expression of the covariance matrix of $\hat{\textbf{g}}$ written as
\begin{equation}
    (\boldsymbol{\Sigma}_{\hat{g}})_{i,j} = \left\{
    \begin{array}{ll}
        \sigma^2 & \mbox{if } i=j \\
        \sum_{l=1}^{\epsilon\text{-rank}} s_l u_{i,l} u_{j,l} & \mbox{if } i \neq j
    \end{array}
\right. \text{for } i,j \in \{1,\ldots,N\}.
\end{equation}

Further, We can see that $\textcolor{black}{\mathrm{E}}(\Tilde{g}_{k,r}) = 0$ for all $k \in \{1,\ldots, N\}$ and $r \in \{1,\ldots, R\}$. Additionally, for $i,j \in \{1,\ldots, R\}$ and $m,n \in \{1, \ldots, N\}$, 
\begin{align}
\begin{split}
    (\textcolor{black}{\boldsymbol{\Sigma}}_{\Tilde{G}})^{m,n}_{i,j} &= \textcolor{black}{\mathrm{Cov}}(\Tilde{g}_{m,i}, \Tilde{g}_{n,j})\\
    &=  \left\{
    \begin{array}{ll}
        0 & \mbox{if } m \neq n \\
        \sigma^2  & \mbox{if } i=j \text{ and } m=n \\
        \sum\limits^{\epsilon\text{-rank}}_{l=1} s_l u_{n,l}^2 & \mbox{if } i \neq j \text{ and } m=n,
    \end{array} 
    \right.
    \end{split}
\end{align}
 \noindent where $(\boldsymbol{\Sigma}_{\hat{G}})^{m,n}_{i,j}$ is the entry $(i,j)$ in the block $(m,n)$ of $\boldsymbol{\Sigma}_{\hat{G}}$. Therefore, the result in Proposition \ref{prop:covG_hat_tilde} follows directly from taking the matrix $\boldsymbol{\Sigma}_{k}$, for $k \in \{1,\ldots,N\}$ as
\begin{equation}
    (\boldsymbol{\Sigma}_{k})_{i,j} = \left\{
    \begin{array}{ll}
        \sigma^2 & \mbox{if } i=j \\
        \sum_{l=1}^{\epsilon\text{-rank}} s_l u_{k,l}^2 & \mbox{if } i \neq j
    \end{array}
\right. \text{for } i,j \in \{1,\ldots,R\}.
\end{equation}


\section{} \label{appendixF}
 
\textcolor{black}{
\begin{align}
    \begin{split}
        F_{\Omega_R}(g) &= P(|\hat{g}^{(r)}_k| \leq g, ~ \forall k \in \{1,\ldots, N\}, ~ \forall r \in \{1,\ldots, R\} )\\
        &= \left ( P(|\hat{g}^{(r)}_k| \leq g, ~ \forall 1\leq k \leq N) \right ) ^R\\
        &= \left( F_{\max\{|\hat{g}_1|, |\hat{g}_2|,\ldots, |\hat{g}_N|\}}(g) \right) ^R
    \end{split}
\end{align}
\begin{align}
    \begin{split}
        F_{\Psi_R}(g)  &= P( \max_{1\leq k \leq N }\{ ~\max_{1\leq r \leq R} \{|\Tilde{g}_{k,r}| \} ~\} \leq g)\\
        &=  \prod_{k=1}^{N} P( ~\max_{1\leq r \leq R} \{|\Tilde{g}_{k,r}| \} \leq g) ~~ \text{(given by the independence property in the Remark (\ref{indep}))}\\
        &= \prod_{k=1}^{N} \textcolor{black}{\mathrm{E}}_{\textbf{a}^{(k)}, \textbf{b}^{(k)}} \left[ P( \max_{1\leq r \leq R} \{|\Tilde{g}_{k,r}|~ | (\textbf{a}^{(k)}, \textbf{b}^{(k)})\} \leq g) \right],
    \end{split}
    \label{eq84}
    \end{align}
\noindent where $\textbf{a}^{(k)} = (a_{1,k}, a_{2,k},\ldots, a_{{\epsilon\text{-rank}},k})^{\textcolor{black}{\mathrm{T}}}$,  $\textbf{b}^{(k)} = (b_{1,k}, b_{2,k},\ldots, b_{{\epsilon\text{-rank}},k})^{\textcolor{black}{\mathrm{T}}}$ and $\textcolor{black}{\mathrm{E}}_{\textbf{a}^{(k)}, \textbf{b}^{(k)}}[.]$ is the expectation with respect to the random variable $(\textbf{a}^{(k)},\textbf{b}^{(k)})$. From the expression of $\Tilde{g}_{k,r}$ given in (\ref{def:g_tilde}), we can see that $\Tilde{g}_{k,1}| (\textbf{a}^{(k)}, \textbf{b}^{(k)}), \Tilde{g}_{k,2}| (\textbf{a}^{(k)}, \textbf{b}^{(k)}),\ldots, \Tilde{g}_{k,R}| (\textbf{a}^{(k)}, \textbf{b}^{(k)})$ are independent and identically distributed. Further, for $k \in \{1, \ldots, N\}$, 
\begin{equation}
        \Tilde{g}_{k,r}| (\textbf{a}^{(k)}, \textbf{b}^{(k)}) \sim \mathcal{CN}( \sum_{l=1}^{\epsilon\text{-rank}} \sqrt{s_l} u_{k,l} (a_{l,k}+\textcolor{black}{\mathrm{j}} b_{l,k}), \sigma^2- \sum_{l=1}^{\epsilon\text{-rank}} s_l u_{k,l}^2), ~ \forall r \in \{1,\ldots, R\}.
\end{equation}
\noindent This implies that  $ |\Tilde{g}_{k,r}|~ |(\textbf{a}^{(k)}, \textbf{b}^{(k)})$ is Rician distributed \cite{stuber}, and its CDF can be written as
    \begin{equation}
        F_{|\Tilde{g}_{k,r}|~|(\textbf{a}^{(k)}, \textbf{b}^{(k)})}(g) = 1-Q_1 \left(\frac{\sqrt{2\left(\sum \limits_{l=1}^{\epsilon\text{-rank}} \sqrt{s_l} u_{k,l} a_{l,k}\right)^2+ 2\left(\sum \limits_{l=1}^{\epsilon\text{-rank}} \sqrt{s_l} u_{k,l} b_{l,k}\right)^2}}{\sqrt{\sigma^2 -\sum \limits_{l=1}^{\epsilon\text{-rank}} s_l u_{k, l}^{2}} }, \frac{\sqrt{2} g}{ \sqrt{\sigma^2 -\sum \limits_{l=1}^{\epsilon\text{-rank}} s_l u_{k, l}^{2}} }\right).
    \end{equation}
\noindent Therefore, 
\begin{align}
    \begin{split}
        P( \max_{1\leq r \leq R} \{|\Tilde{g}_{k,r}|~ | (\textbf{a}^{(k)}, \textbf{b}^{(k)}) \} \leq g) 
        &= \prod_{r=1}^{R} P( |\Tilde{g}_{k,r}|~ | (\textbf{a}^{(k)}, \textbf{b}^{(k)}) \leq g) ~ \text{(independence)}\\
        &= \left[ P( |\Tilde{g}_{k,r}|~ | (a^{(k)}, b^{(k)}) \leq g)\right]^R ~\text{(identical distribution)} \\
        &= \left[ F_{|\Tilde{g}_{k,r}|~|(\textbf{a}^{(k)}, \textbf{b}^{(k)})}(g) \right] ^R.
    \end{split}
\end{align}
\noindent Furthermore, 
\begin{equation}
    f_{\textbf{a}^{(k)},\textbf{b}^{(k)}} (a_{1,k},\ldots, a_{\epsilon\text{-rank},k}, b_{1,k},\ldots, b_{\epsilon\text{-rank},k}) = \prod_{l=1}^{\epsilon\text{-rank}} \frac{1}{\pi} \exp{(-(a^2_{l,k}+b^2_{l,k}))}.
\end{equation}
\noindent Therefore, we can write 
\begin{align}
    \begin{split}
        &\textcolor{black}{\mathrm{E}}_{\textbf{a}^{(k)}, \textbf{b}^{(k)}} \left[ P( \max_{1\leq r \leq R} \{|\Tilde{g}^{(r)}_k|~ | (\textbf{a}^{(k)}, \textbf{b}^{(k)})\} \leq g) \right] = \idotsint\limits_{-\infty}^{\infty} \prod_{l=1}^{\epsilon\text{-rank}} \frac{1}{\pi} \exp{(-(a^2_{l,k}+b^2_{l,k}))} \\ &\left[ 1-Q_1 \left(\frac{\sqrt{2\left(\sum \limits_{l=1}^{\epsilon\text{-rank}} \sqrt{s_l} u_{k,l} a_{l,k}\right)^2+ 2\left(\sum \limits_{l=1}^{\epsilon\text{-rank}} \sqrt{s_l} u_{k,l} b_{l,k}\right)^2}}{\sqrt{\sigma^2 -\sum \limits_{l=1}^{\epsilon\text{-rank}} s_l u_{k, l}^{2}} }, \frac{\sqrt{2} g}{ \sqrt{\sigma^2 -\sum \limits_{l=1}^{\epsilon\text{-rank}} s_l u_{k, l}^{2}} }\right) \right] ^R \\ &da_{1,k}, \ldots, da_{\epsilon\text{-rank},k}, db_{1,k}, \ldots, db_{\epsilon\text{-rank},k}\\
        &= \textcolor{black}{\mathrm{E}}_{(\textbf{a},\textbf{b})} \left[ \left[ 1-Q_1 \left(\frac{\sqrt{2\left(\sum \limits_{l=1}^{\epsilon\text{-rank}} \sqrt{s_l} u_{k,l} a_{l}\right)^2+ 2\left(\sum \limits_{l=1}^{\epsilon\text{-rank}} \sqrt{s_l} u_{k,l} b_{l}\right)^2}}{\sqrt{\sigma^2 -\sum \limits_{l=1}^{\epsilon\text{-rank}} s_l u_{k, l}^{2}} }, \frac{\sqrt{2} g}{ \sqrt{\sigma^2 -\sum \limits_{l=1}^{\epsilon\text{-rank}} s_l u_{k, l}^{2}} }\right) \right] ^R \right], 
    \end{split}
\end{align}
\noindent because $(\textbf{a}^{(k)},\textbf{b}^{(k)})$ is identically distributed as $(\textbf{a},\textbf{b})$.
\noindent On the other hand, we define the random variable $Z_k$ as: 
    \begin{align}
        \begin{split}
            Z_k &= \left( \sum \limits_{l=1}^{\epsilon\text{-rank}} \sqrt{s_l} u_{k,l} a_l\right)^2+ \left(\sum \limits_{l=1}^{\epsilon\text{-rank}} \sqrt{s_l} u_{k,l} b_l\right)^2 = Z^{(a)}_k + Z^{(b)}_k, 
        \end{split}
    \end{align}
    \noindent where $Z^{(a)}_k = \textbf{a}^{\textcolor{black}{\mathrm{T}}} \boldsymbol{\gamma}_k^{\textcolor{black}{\mathrm{T}}} \boldsymbol{\gamma}_k \textbf{a} $, $Z^{(b)}_k = \textbf{b}^{\textcolor{black}{\mathrm{T}}} \boldsymbol{\gamma}_k^{\textcolor{black}{\mathrm{T}}} \boldsymbol{\gamma}_k \textbf{b}$ and $\boldsymbol{\gamma}_k = (\sqrt{s_1} u_{k,1}, \sqrt{s_2} u_{k,2},\ldots,\sqrt{s_{\epsilon\text{-rank}}} u_{k,\epsilon\text{-rank}})$. To determine the distribution of $Z_k$, we have that 
    \begin{align}
    \begin{split}
        Z^{(a)}_k &= \textbf{a}^{\textcolor{black}{\mathrm{T}}} \boldsymbol{\gamma}_k^{\textcolor{black}{\mathrm{T}}} \boldsymbol{\gamma}_k \textbf{a} = \boldsymbol{\gamma}_k \boldsymbol{\gamma}^T_k  \left( \textbf{a}^{\textcolor{black}{\mathrm{T}}} \frac{\boldsymbol{\gamma}_k^{\textcolor{black}{\mathrm{T}}} \boldsymbol{\gamma}_k }{\boldsymbol{\gamma}_k \boldsymbol{\gamma}^T_k} \textbf{a} \right) = (\sqrt{\boldsymbol{\gamma}_k \boldsymbol{\gamma}^T_k} \textbf{a}^{\textcolor{black}{\mathrm{T}}}) \textbf{P}_k  (\sqrt{\boldsymbol{\gamma}_k \boldsymbol{\gamma}^T_k} \textbf{a}^{\textcolor{black}{\mathrm{T}}})^{\textcolor{black}{\mathrm{T}}}, 
    \end{split}
    \end{align}
    where $\textbf{P}_k = \frac{\boldsymbol{\gamma}_k^{\textcolor{black}{\mathrm{T}}} \boldsymbol{\gamma}_k }{\boldsymbol{\gamma}_k \boldsymbol{\gamma}^T_k} $ is a projector of rank 1. Therefore, $Z^{(a)}_k \sim \chi^2 (1)$ and $f_{Z^{(a)}_k } (z) = \frac{1}{\sqrt{\pi \boldsymbol{\gamma}_k \boldsymbol{\gamma}^T_k z}}\exp(-\frac{z}{\boldsymbol{\gamma}_k \boldsymbol{\gamma}^T_k}) $. 
    Therefore, the pdf of $Z_k$ can be written as: 
    \begin{align}
        f_{Z_k}(z) &= f_{Z^{(a)}_k } (z) * f_{Z^{(b)}_k } (z) = \frac{1}{\boldsymbol{\gamma}_k \boldsymbol{\gamma}^T_k }\exp{\left(-\frac{z}{\boldsymbol{\gamma}_k \boldsymbol{\gamma}^T_k} \right)}, ~~ z>0.
    \end{align}
    Finally, we can write the CDF of $\Psi_R$ as 
    \begin{align}
        \begin{split}
            F_{\Psi_R}(g) &= \prod_{k=1}^{N} \textcolor{black}{\mathrm{E}}_{\textbf{a}^{(k)}, \textbf{b}^{(k)}} \left[ P( \max_{1\leq r \leq R} \{|\Tilde{g}_{k,r}|~ | (\textbf{a}^{(k)}, \textbf{b}^{(k)})\} \leq g) \right] ~ \text{ (by (\ref{eq84}))} \\
            &= \prod_{k=1}^{N} \textcolor{black}{\mathrm{E}}_{Z_k} \left[ \left[ 1-Q_1 \left(\frac{\sqrt{2 Z_k}}{\sqrt{\sigma^2 -\sum \limits_{l=1}^{\epsilon\text{-rank}} s_l u_{k, l}^{2}} }, \frac{\sqrt{2} g}{ \sqrt{\sigma^2 -\sum \limits_{l=1}^{\epsilon\text{-rank}} s_l u_{k, l}^{2}} }\right) \right] ^R \right]\\
            &= \prod_{k=1}^N  \int_{0}^{+\infty} \frac{1}{\sum\limits_{l=1}^{\epsilon\text{-rank}} s_l u_{k, l}^{2}} \exp \left({-\frac{r}{\sum\limits_{l=1}^{\epsilon\text{-rank}} s_l u_{k, l}^{2}}}\right)
         \left( 1-Q_1 \left(\frac{\sqrt{2r}}{\sqrt{\sigma^2 -\sum \limits_{l=1}^{\epsilon\text{-rank}} s_l u_{k, l}^{2}} }, \frac{\sqrt{2}g}{\sqrt{\sigma^2 -\sum \limits_{l=1}^{\epsilon\text{-rank}} s_l u_{k, l}^{2}} }\right) \right)^R dr.
        \end{split}
    \end{align}}
    
   

\section{} \label{appendixG}

By the definitions of the matrices $\boldsymbol{\Sigma}_{\Tilde{G}} (R)$ and $\boldsymbol{\mathbb{I}}(R)$, we can write 
\begin{align}
    \begin{split}
        \norm{ \sigma^2 \boldsymbol{\mathbb{I}}(R)-\boldsymbol{\Sigma}_{\Tilde{G}} (R) }_1 &= \max_{1 \leq k \leq N} \norm{ \sigma^2 \mathbf{1}_{R \times R} - \boldsymbol{\Sigma}_k}_1 \\
        &= \max_{1 \leq k \leq N} (R-1) (\sigma^2 - \sum_{l=1}^{\epsilon\text{-rank}} s_l u_{k,l}^2)\\
        & \leq \max_{1 \leq k \leq N} N(\sigma^2 - \sum_{l=1}^{\epsilon\text{-rank}} s_l u_{k,l}^2) (\text{ for } R \leq N)\\
        & \leq \max_{1 \leq k \leq N} N(\sum_{l=1}^{N} s_l u_{k,l}^2- \sum_{l=1}^{\epsilon\text{-rank}} s_l u_{k,l}^2) \text{ (because } \sum_{l=1}^{N} s_l u_{k,l}^2 = \sigma^2 ) \\
        & \leq \max_{1 \leq k \leq N} N \sum_{l=1+\epsilon\text{-rank}}^{N} s_l u_{k,l}^2\\
        & \leq \max_{1 \leq k \leq N} N \sum_{l=1+\epsilon\text{-rank}}^{N} \epsilon~ u_{k,l}^2 \text{ (by definition of $\epsilon\text{-rank}$})\\
        & \leq N \epsilon \text{ (because} \sum_{l=1+\epsilon\text{-rank}}^{N} u_{k,l}^2 < 1).
    \end{split}
\end{align}
\noindent Therefore, for $\epsilon = \frac{\epsilon'}{2N}$,
\begin{equation}
    \norm{ \sigma^2 \boldsymbol{\mathbb{I}}(R)-\boldsymbol{\Sigma}_{\Tilde{G}} (R) }_1 \leq \frac{\epsilon '}{2}.
    \label{eq100}
\end{equation}

\noindent On the other hand, by the definition of $\boldsymbol{\Sigma}_G$ and $\boldsymbol{\Sigma}_{\hat{G}}$, we can write 
\begin{align}
    \begin{split}
        \norm{\boldsymbol{\Sigma}_{\hat{G}} - \boldsymbol{\Sigma}_G}_1 &= \norm{\boldsymbol{\Sigma}_{\hat{g}} - \boldsymbol{\Sigma}_g}_1.
    \end{split}
    \label{eq101}
\end{align}
\noindent Moreover, for $i,j \in \{1,\ldots,N\}$, we have 
\begin{align}
   (\boldsymbol{\Sigma}_g)_{i,j} &= \sum_{l=1}^{N} s_i u_{i,l} u_{j,l},\\
   \begin{split}
   (\boldsymbol{\Sigma}_{\hat{g}})_{i,j} 
    &= \left\{
    \begin{array}{ll}
        \sum_{l=1}^{N} s_i u_{i,l}^2  & \mbox{if } i=j \\
        \sum_{l=1}^{\epsilon\text{-rank}} s_i u_{i,l} u_{j,l}& \mbox{if } i \neq j
    \end{array}.
    \right.
    \end{split}
\end{align}
We define the matrix $\boldsymbol{\Sigma}'$ such that for $i,j \in \{1,\ldots,N\}$, 
\begin{equation}
    (\boldsymbol{\Sigma}')_{i,j} = \sum_{l=1}^{\epsilon\text{-rank}} s_i u_{i,l} u_{j,l}.
\end{equation}
We consider the max norm $\norm{.}_{max}$ and the spectral norm $\norm{.}_2$ defined for a matrix $\textbf{A}$ of size $N \times N$ as follows: 
\begin{equation}
\begin{array}{ll}
    \norm{\textbf{A}}_{max} = \underset{1 \leq i,j \leq N}{\max} \lvert a_{i,j} \rvert, ~~ & \norm{\textbf{A}}_2 = \sigma_{max}(\textbf{A}).  \\
    \end{array} 
\end{equation}
where $\sigma_{max}(\textbf{A})$ is the maximum singular value of $\textbf{A}$. 
we can see that 
\begin{equation}
    \norm{\boldsymbol{\Sigma}_g - \boldsymbol{\Sigma}_{\hat{g}}}_{max} \leq \norm{\boldsymbol{\Sigma}_g -\boldsymbol{\Sigma}'}_{max},
\end{equation}
because while $\boldsymbol{\Sigma}_g - \boldsymbol{\Sigma}_{\hat{g}}$ and $\boldsymbol{\Sigma}_g - \boldsymbol{\Sigma}'$ have the same off-diagonal elements, $\boldsymbol{\Sigma}_g - \boldsymbol{\Sigma}_{\hat{g}}$ has zeros as diagonal elements and $\boldsymbol{\Sigma}_g - \boldsymbol{\Sigma}' =  \sum_{l=R+1}^{N} s_l \textbf{u}_{l} \textbf{u}_{l}^{\textcolor{black}{\mathrm{T}}} $ has positive diagonal elements. Therefore, 
\begin{align}
\begin{split}
    \norm{\boldsymbol{\Sigma}_g - \boldsymbol{\Sigma}_{\hat{g}}}_{max} &\leq \norm{\boldsymbol{\Sigma}_g - \boldsymbol{\Sigma}'}_{max}\\
    &\leq \norm{\sum_{l=1+\epsilon\text{-rank}}^{N} s_l \textbf{u}_{l} \textbf{u}_{l}^{\textcolor{black}{\mathrm{T}}}}_{max}\\
    &\leq \norm{\sum_{l=1+\epsilon\text{-rank}}^{N} s_l \textbf{u}_{l} \textbf{u}_{l}^{\textcolor{black}{\mathrm{T}}}}_{2}~ \text{ (matrix norm equivalence \cite{wiki2})} \\
    &\leq s_{1+\epsilon\text{-rank}}\\
    &\leq \epsilon.
    \end{split}
\label{eq107}
\end{align}

\noindent This means that, 
\begin{align}
    \norm{\boldsymbol{\Sigma}_{\hat{G}} - \boldsymbol{\Sigma}_G}_1 &= \norm{\boldsymbol{\Sigma}_{\hat{g}} - \boldsymbol{\Sigma}_g}_1 \\
    & \leq N \epsilon \text{ (using the result from (\ref{eq107})).}
\end{align}
\noindent Therefore, for $\epsilon = \frac{\epsilon}{2N}$, 
\begin{equation}
    \norm{\boldsymbol{\Sigma}_{\hat{G}} - \boldsymbol{\Sigma}_G}_1 \leq \frac{\epsilon '}{2}.
    \label{eq110}
\end{equation}

\noindent By combining equations (\ref{eq110}) and (\ref{eq100}), we get

\begin{align}
    \left| \norm{\boldsymbol{\Sigma}_{\hat{G}}(R) - \boldsymbol{\Sigma}_{\Tilde{G}}(R)}_1 - \norm{\boldsymbol{\Sigma}_{G}(R) - \sigma^2 \boldsymbol{\mathbb{I}}(R)}_1  \right| & \leq 
     \norm{\boldsymbol{\Sigma}_{\hat{G}}(R) - \boldsymbol{\Sigma}_{\Tilde{G}}(R) - \boldsymbol{\Sigma}_{G}(R) + \sigma^2 \boldsymbol{\mathbb{I}}(R)}_1\\ & \leq 
     \norm{\boldsymbol{\Sigma}_{\hat{G}}(R) - \boldsymbol{\Sigma}_G (R)}_1 + \norm{ \sigma^2 \boldsymbol{\mathbb{I}}(R)-\boldsymbol{\Sigma}_{\Tilde{G}} (R) }_1 \\ & \leq 
     \frac{\epsilon '}{2} + \frac{\epsilon '}{2} \\ &\leq 
     \epsilon '.
\end{align}

\section{} \label{appendixH}

We have 
\begin{align}
        \boldsymbol{\Sigma}_{G} (R) &= 
        \begin{pmatrix}
        \boldsymbol{\Sigma}_{{g}} & 0 & \cdots & 0 \\
        0 & \boldsymbol{\Sigma}_{{g}} & \cdots & 0 \\
        \vdots  & \vdots  & \ddots & \vdots  \\
        0 & 0 & \cdots & \boldsymbol{\Sigma}_{{g}}
        \end{pmatrix} & 
        \boldsymbol{\mathbb{I}}(R) = 
        \begin{pmatrix}
        \boldsymbol{\mathbf{1}}_{R \times R} & 0 & \cdots & 0 \\
        0 & \boldsymbol{\mathbf{1}}_{R \times R} & \cdots & 0 \\
        \vdots  & \vdots  & \ddots & \vdots  \\
        0 & 0 & \cdots & \boldsymbol{\mathbf{1}}_{R \times R}
    \end{pmatrix},
    \end{align}

\noindent We can see that $\boldsymbol{\Sigma}_{G}(R)$ is a block diagonal matrix with $R$ equal blocks, each of size $N \times N$. On the other hand, $\boldsymbol{\mathbb{I}} (R)$ is a block matrix with $N$ equal blocks, each of size $R \times R$. Then, both matrices $\boldsymbol{\Sigma}_{G}(R)$ and $\boldsymbol{\mathbb{I}} (R)$ are of size $NR \times NR$. Furthermore, for $R$ divisor of $N$, $\boldsymbol{\Sigma}_{G}(R) - \sigma^2 \boldsymbol{\mathbb{I}} (R)$ is block diagonal with $R$ equal blocks, each of size $N \times N$. More specifically, one of the diagonal blocks of the matrix $\boldsymbol{\Sigma}_{G}(R) - \sigma^2 \boldsymbol{\mathbb{I}} (R)$, say the first block, is written as
\begin{equation}
    (\boldsymbol{\Sigma}_{G}(R) - \sigma^2 \boldsymbol{\mathbb{I}} (R))^{1,1} = \boldsymbol{\Sigma}_g - \sigma^2 diag\underbrace{( \mathbf{1}_{R \times R}, \ldots, \mathbf{1}_{R \times R})}_\text{ $\frac{N}{R}$ times}.
\end{equation}

\noindent In other words, by examining one of the diagonal blocks of the matrix $(\boldsymbol{\Sigma}_{G}(R) - \sigma^2 \boldsymbol{\mathbb{I}} (R))$, we can see that $R$ determines the size of diagonal blocks of $\textcolor{black}{\boldsymbol{\Sigma}}_g$ from which $\sigma^2$ is subtracted.

\noindent Further, we have 
\begin{equation}
    (\boldsymbol{\Sigma}_g)_{k,l} = \sigma^{2} J_{0}\left(\frac{2 \pi(k-\ell)}{N-1} W\right), \text{ for } \ell ,k \in \{1,\ldots,N\}.
\end{equation}

\noindent Therefore, if we denote by $(\boldsymbol{\Sigma}_{G}(R) - \sigma^2 \boldsymbol{\mathbb{I}} (R))_{k,\ell}^{m,n}$ the entry $(k,\ell)$ in the block $(m,n)$ of the matrix $(\boldsymbol{\Sigma}_{G}(R) - \sigma^2 \boldsymbol{\mathbb{I}} (R))$ with $ 1 \leq m,n \leq R$ and $ 1 \leq k,\ell \leq N $, then we can write 
\begin{equation}
    (\boldsymbol{\Sigma}_{G}(R) - \sigma^2 \boldsymbol{\mathbb{I}} (R))_{k,\ell}^{m,n}  = \left\{
    \begin{array}{ll}
        \sigma^2 \left(J_{0}\left(\frac{2 \pi(k-\ell)}{N-1} W\right) - 1 \right) & \mbox{if } m=n \text{ and } (k,\ell) \in S\\
        \sigma^2 J_{0}\left(\frac{2 \pi(k-\ell)}{N-1} W\right) & 
       \mbox{if } m=n \text{ and } (k,\ell) \not\in S\\
        0 & \text{ otherwise,}
    \end{array}
    \right.
\end{equation}
where $S$ is defined as 
$$
S = \{ (k,\ell)\in \mathbb{N}, \exists p \in \{0,\ldots,\frac{N}{R}-1\},\text{ such that } pR+1 \leq k,\ell \leq (p+1)R \}.
$$

\noindent Therefore, the norm 1 of the matrix $(\boldsymbol{\Sigma}_{G}(R) - \sigma^2 \boldsymbol{\mathbb{I}} (R))$ is equal to the norm 1 of one of its blocks (all its blocks are equal) and we can write
\begin{align}
\begin{split}
     &\norm{\boldsymbol{\Sigma}_{G}(R) - \sigma^2 \boldsymbol{\mathbb{I}} (R)}_1  = \norm{(\boldsymbol{\Sigma}_{G}(R) - \sigma^2 \boldsymbol{\mathbb{I}} (R))^{1,1}}_1  \\
    &= \norm{\boldsymbol{\Sigma}_g - \sigma^2 diag\underbrace{( \mathbf{1}_{R \times R}, \ldots, \mathbf{1}_{R \times R})}_\text{ $\frac{N}{R}$ times}}_1\\
    &= \max_{ 0 \leq p \leq \frac{N}{R}-1} ~~ \max_{pR+1 \leq l \leq (p+1)R}  \sum_{k=1}^{pR} \sigma^2 \left | J_{0}\left(\frac{2 \pi(k-\ell)}{N-1} W\right)  \right |  \\ & + \sum_{k=pR+1}^{(p+1)R} \sigma^2 \left | J_{0}\left(\frac{2 \pi(k-\ell)}{N-1} W\right) - 1 \right | \\ & + \sum_{k=pR+R+1}^{N}  \sigma^2 \left| J_{0}\left(\frac{2 \pi(k-\ell)}{N-1} W\right) \right|.
\end{split}
\end{align}

\noindent To solve (\ref{P3}), we would like to find $R_3^*$, a divisor of $N$, that minimizes  $\norm{\boldsymbol{\Sigma}_{G}(R) - \sigma^2 \boldsymbol{\mathbb{I}} (R)}_1 $ given by the expression above. We can see that increasing $R$ increases the size of the diagonal blocks of $\boldsymbol{\Sigma}_g$ from which $\sigma^2$ is subtracted. However, this does not necessarily decrease the maximum absolute column sum of the matrix $\boldsymbol{\Sigma}_g - \sigma^2 diag\underbrace{( \mathbf{1}_{R \times R}, \ldots, \mathbf{1}_{R \times R})}_\text{ $\frac{N}{R}$ times}$. In fact, increasing $R$ decreases the norm 1 of the matrix $\boldsymbol{\Sigma}_g - \sigma^2 diag\underbrace{( \mathbf{1}_{R \times R}, \ldots, \mathbf{1}_{R \times R})}_\text{ $\frac{N}{R}$ times}$, if and only if, subtracting $\sigma^2$ from the entries of the diagonal blocks of $\boldsymbol{\Sigma}_g$ decreases their absolute value. Now, since the entries of $\boldsymbol{\Sigma}_g$ are the highest on the diagonal (\textit{i.e.} $\sigma^2$), then they start decreasing with the off-diagonals. Therefore, subtracting $\sigma^2$ from the entries of the diagonal blocks of $\boldsymbol{\Sigma}_g$ decreases their absolute value, if and only if, they are greater than $\frac{\sigma^2}{2}$. Consequently, we start decreasing the the maximum absolute column sum of the matrix $\boldsymbol{\Sigma}_g - \sigma^2 diag\underbrace{( \mathbf{1}_{R \times R}, \ldots, \mathbf{1}_{R \times R})}_\text{ $\frac{N}{R}$ times}$ by increasing $R$ until we reach $ \sigma^2 J_{0}\left(\frac{2 \pi(R-1)}{N-1} W\right) = \frac{\sigma^2}{2}$. More formally,

\noindent If $R_3^*$ is the greatest divisor of $N$ verifying 
\begin{equation}
    J_{0}\left(\frac{2 \pi(R_3^*-1)}{N-1} W\right) \leq 0.5,
\end{equation}
\noindent then
\begin{equation}
     \norm{\boldsymbol{\Sigma}_{G}(R_3^*) - \sigma^2 \boldsymbol{\mathbb{I}} (R_3^*)}_1  \leq \norm{\boldsymbol{\Sigma}_{G}(R) - \sigma^2 \boldsymbol{\mathbb{I}} (R)}_1  \text{ for all $R$ divisor of $N$}. 
\end{equation}


\bibliographystyle{IEEEtran}
\bibliography{references.bib}

%

\end{document}